\documentclass[final]{siamltex}

\usepackage{hyperref,amssymb,booktabs,multirow,graphicx}

\def\1{{\sf 1}}
\def\A{{\mathcal{A}}}
\def\D{{\mathcal{D}}}
\def\L{{\mathcal{L}}}
\def\J{{\mathcal{J}}}
\def\P{{\mathcal{P}}}
\def\Q{{\mathbf{Q}}}

\def\<{\langle}
\def\>{\rangle}
\def\={{\!\!\!=\!\!\!}}

\def\e{{\mathbf{e}}}

\newcommand{\field}[1]{\ensuremath{\mathbb{#1}}}

\newcommand{\R}{\field{R}}

\newcommand{\norm}[1]{\ensuremath{\left|\!\left|#1\right|\!\right|}}

\newcommand{\levy}{L\'{e}vy~}

\newcommand{\levyito}{L\'{e}vy-It\^{o}}

\newcommand{\garding}{G{\aa}rding~}

\newcommand{\dsp}{\displaystyle}
\newcommand{\be}{\begin{equation}}
\newcommand{\ee}{\end{equation}}
\newcommand{\bea}{\begin{eqnarray}}
\newcommand{\eea}{\end{eqnarray}}
\newcommand{\eee}{\end{eqnarray*}}
\newcommand{\bee}{\begin{eqnarray*}}

\newtheorem{remark}[theorem]{Remark}

\title{Pricing Two-asset Options under Exponential L\'{e}vy Model Using a Finite Element Method}

\author{Xun Li\thanks{Department of Applied Mathematics, Hong Kong Polytechnic University, Hung Hom, Kowloon, Hong Kong (malixun@polyu.edu.hk). This author acknowledges financial support from General Research Fund of Hong Kong SAR no. 15209614 and 15224215.} \and Ping Lin\thanks{Department of Mathematics, University of Dundee, Dundee DD1 4HN, Scotland, UK (plin@maths.dundee.ac.uk). This author acknowledges financial support from NSF of China No 91430106, Fundamental Research Funds for the Central Universities Nos. 06108037 and FRF-BR-13-023.}  \and Xue-Cheng Tai $^\S$\thanks{Department of Mathematics, University of Bergen, Johannes Brunsgate 12, Bergen 5008, Norway (tai@cma.uio.no).} \and Jinghui Zhou\thanks{Division of Mathematical Sciences, School of Physical \& Mathematical Sciences, Nanyang Technological University, 21 Nanyang Link, Singapore 637371 (xctai@ntu.edu.sg, jhzhou@ntu.edu.sg). The research has been supported by SUG 20/07  of Nanyang Technological University.  In addition, support from MOE (Ministry of Education) Tier II project T207N2202 and IDM project NRF2007IDM-IDM002-010 is also gratefully acknowledged.}}

\begin{document}

\maketitle

\begin{abstract}
This article presents a finite element method (FEM) for a partial integro-differential equation (PIDE) to price two-asset options with underlying price processes modeled by an exponential L\'{e}vy process. We provide a variational formulation in a weighted Sobolev space, and establish existence and uniqueness of the FEM-based solution. Then we discuss the localization of the infinite domain problem to a finite domain and analyze its error. We tackle the localized problem by an explicit-implicit time-discretization of the PIDE, where the space-discretization is done through a standard continuous finite element method. Error estimates are given for the fully discretized localized problem where two assets are assumed to have uncorrelated jumps. Numerical experiments for the polynomial option and a few other two-asset options shed light on good performance of our proposed method. 
\end{abstract}



\pagestyle{myheadings}
\thispagestyle{plain}
\markboth{X. LI, P. LIN, X.-C. TAI and J.H. ZHOU}{PRICING TWO-ASSET OPTIONS UNDER FEM}

\section{Introduction}
\noindent
It is well documented that there are disadvantages for diffusion option pricing models (cf. \cite{BlackScholes:1973,Merton:1973,Dupire:1994,Derman:1994,HullWhite:1987,Heston:1993,Hagan:2002}) to capture the risk when abnormal market movements exist. Recent empirical studies and high-impact market crash show that the property of dramatic fluctuation in asset dynamics should be incorporated into diffusion models (cf. \cite{Belomestny:2012}). On the other hand, the inclusion of jumps into asset price modeling has been developed for many years. Merton (cf. \cite{Merton:1976}) originally introduced Poisson jump process into diffusion models. The Poisson jump diffusion model is an example of exponential \levy models (ELM) where the underlying price dynamics is represented as an exponential  of a \levy process (cf. \cite{David:2004,Sato:2002,FL:2008,Toivanen:2008,KL:2009,JP:2010,GM:2012}).  The extension of diffusion models to exponential \levy models allows to calibrate the models to the market price of options and to reproduce various implied volatility skews/smiles.

Conventionally the valuation of an option under a diffusion model (or Black-Scholes-Merton framework) requires to solve a parabolic partial differential equation. Detailed treatments could be found in \cite{Wilmott:1993, Wilmott:1998}. Assuming that systematic risk can be diversified under exponential \levy models, we may express the option price in terms of the solution of a parabolic integro-differential equation (PIDE), which includes a second order differential operator and a nonlocal integral operator. Option pricing with exponential \levy models has been studied in recent literature such as \cite{Geman:2002,Schoutens:2003,Rama:2003}.

A variety of finite difference methods to solve the one dimensional PIDE have been proposed  recently (cf. \cite{FujiwaraKunita:1985, Admin:1993,Zhang:1997,Briani:2004,Amadori:2000,AndersenAndreasen:2000,Forsyth:2004,Forsyth:2005,RamaCont:2005,Forsyth:2008,Pang:2008}). As an equivalence, binomial lattice methods  were adopted in Admin (cf. \cite{Admin:1993}) for one dimensional jump-diffusion models with a finite jump intensity. Andersen and Andreasen (cf. \cite{AndersenAndreasen:2000}) proposed an operator splitting method where the local part is treated with an implicit step and the nonlocal part is coped with an explicit scheme. Zhang (cf. \cite{Zhang:1997}) developed a semi-implicit finite difference scheme for a jump-diffusion model for pricing the American option. Although finite difference method is relatively efficient for one dimensional pricing problem, finite element method could provide a more general approach for tackling pricing problems on several assets.

In addition to the finite difference method, recently the variational formulation has been introduced by Matache et. al. (cf. \cite{Schwab:2004,Schwab:2005a,Schwab:2005b}) to the one-dimensional PIDE. These papers provided a rigorous analysis of consistency, stability and convergence for a wavelet Galerkin finite element method.
Their analysis based on a weighted Sobolev space provides a good tool for the PIDE in an unbounded domain. Using the finite element method for option pricing models is extensively discussed in Topper (cf. \cite{Topper:2005}). In comparison to finite difference methods, it has  several advantages in coping with domain geometry, boundary conditions and solution smoothness. Pricing multi-asset options under exponential \levy models involves several techniques: smoothing initial and boundary conditions, coping with possibly irregular mapped domains, localizing an unbounded domain to a bounded domain, treating possible singularity associated with certain jumps, discretizing the equation in both temporal and spatial variables. There are a few papers discussing the pricing for two asset option within the framework of ELM. The two asset option with jumps was priced or modeled through a Markov chain approach (cf. \cite{Martzoukos:2003}). Forsyth et.al. (cf. \cite{Forsyth:2008}) use a finite difference scheme for the two dimensional PIDE, but a theoretical error analysis of the PIDE is absent
 in their studies.

In this paper, we present a finite element method for the exponential \levy two-asset option pricing or the two dimensional PIDE. We will consider both put and call options.
Some details of the exponential \levy model and pricing equation for the two-asset option are presented in section 2. In addition smoothed initial and boundary conditions will be constructed and the error between the smoothed problem and the original problem will be estimated in this section. Section 3 begins with a variational formulation of the PIDE in a weighted Sobolev space, and existence and uniqueness of its solution. Then we localize the infinite domain to a bounded domain. In the bounded domain an explicit-implicit time-discretization combined with a continuous finite element method is introduced for the PIDE and error analysis is also done where two assets are assumed to have uncorrelated jumps.
Numerical experiments are given in Section 4 for the polynomial option and other two-asset options.

\section{Exponential \levy model for two assets}
\noindent In exponential \levy models, the stochastic dynamics of risky assets $S(\cdot)=(S_1(\cdot),S_2(\cdot))$ with initial prices $(S_1^0,S_2^0)$ is represented as the exponential of a \levy process:
\[S_i(t)=S_i^0 e^{r t+X_i(t)},\quad i=1,2,\]
where $r$ is risk-free interest rate and $X(t)=\left( X_1(t),X_2(t) \right)$ is a two dimensional \levy process starting from $0$ under risk-neutral probability $\Q$. The details for \levy process could be found in \cite{David:2004,Rama:2003}. The absence of arbitrage imposes that the discounted prices $e^{-rt}S(t)$ is a martingale under such measure $\Q$.

In the context of \levyito~decomposition (cf. pp.119-135 in \cite{Sato:2002}), the risk neutral dynamics of $S_i(\cdot)$ is given by
\begin{eqnarray*}
S_i(t)&=&S_i^0+\int_{0}^t r S_i(u-)du + \int_{0}^t \sigma_i S_i(u-)dW_i(u)\\
&&+\int_{0}^{t}\int_{\R}(e^{y_i}-1)S_i(u-)J_{X_i}(du\cdot dy_i),
\end{eqnarray*}
where $\sigma_i$ is the diffusion volatility for $S_i(\cdot)$, $W_1(\cdot),W_2(\cdot)$ are standard Brownian motions correlated by $\rho$ and $J_{X_i}$ is the compensated measure describing the jumps of $X_i$. In the following, we will consider the pricing equation for two-asset option under exponential \levy model and the initial and boundary conditions for the equation.

\subsection{Integro-differential equation for two-asset option}
In the classic martingale pricing approach from the insights of \cite{BlackScholes:1973,Merton:1973,Merton:1976,Rama:2003}, the value of a European option is defined as a discounted conditional expectation of its terminal payoff under a risk-neutral probability $\Q$. Following the ideas in \cite{RamaCont:2005,Schwab:2004,Schwab:2005a} for one dimensional case and assuming the jump components are independent, we formulate the pricing problem of a European-style two-asset option at time $t$ with strike price $K$, maturity $T$ and payoff $H(\cdot,\cdot)$ as the following parabolic integro-differential equation
\begin{eqnarray}\label{eq:PIDE}
&&\left\{\begin{array}{l} \dsp \frac{\partial V}{\partial t} + rS_1\frac{\partial V}{\partial S_1} + rS_2\frac{\partial V}{\partial S_2} - rV \\ [3mm]
\dsp + \frac{1}{2} \sigma_1^2 S_1^2 \frac{\partial^2 V}{\partial S_1^2} + \rho \sigma_1 \sigma_2 S_1 S_2 \frac{\partial^2 V}{\partial S_1 \partial S_2}+\frac{1}{2} \sigma_2^2 S_2^2 \frac{\partial ^2 V}{\partial S_2^2}   \\[3mm]
\dsp+ \int_{\mathbb{R}}\left(V(t,S_1e^{y_1},S_2)-V(t,S_1,S_2)-(e^{y_1}-1)S_1\frac{\partial V}{\partial S_1}\right) \nu_1(dy_1) \\ [3mm]
\dsp+\int_{\mathbb{R}}\left(V(t,S_1,S_2e^{y_2})-V(t,S_1,S_2)-(e^{y_2}-1) S_2\frac{\partial V}{\partial S_2}\right) \nu_2(dy_2) = 0, \\ [5mm]
V(T,S_1,S_2) = H(S_1,S_2),
\end{array}\right.
\end{eqnarray}
where $\nu_i(dy_i)$ is the \levy measure to describe the activity of jump size $y_i$ for underlying asset $S_i$.
The regularity properties for $V(t,S_1,S_2)$ and the detailed derivation of the PIDE (\ref{eq:PIDE}) could be found in \cite{Zhou:2009}. The payoff function $H(\cdot,\cdot)$ of certain two-asset options  are given in Table \ref{tab:2DPayoffs}.
\begin{table}[htbp]
\centering
\begin{tabular}{ll}
\toprule
  \textbf{Option Type}   & \textbf{Payoff $H(S_1,S_2)$} \\ \midrule
  Basket Call            & $\max\left((w_1S_1+w_2S_2)-K,0\right)$   \\
  Basket Put             & $\max\left(K-(w_1S_1+w_2S_2),0\right)$  \\
  Worse of 2 Assets      & $\max\left(\min(S_1,S_2),0\right)$   \\
  Minimum of 2 Assets    & $\max\left(K-\min(S_1,S_2),0\right)$   \\
\bottomrule
\end{tabular}
\caption{Payoff functions for some two-asset options.}
\label{tab:2DPayoffs}
\end{table}
Before we solve the problem (\ref{eq:PIDE}) by using numerical method, we introduce variable transformations and specify \levy measure (See \cite{Merton:1973}) as follow:
\begin{eqnarray}\label{eq:var_transforamtion}
\left\{\begin{array}{l}
\tau=T-t, x_i=\ln S_i, \\ [2mm]
\nu_i(dy_i)=k_i(y_i)dy_i, \\ [2mm]
k_i(y_i)=\lambda_i \frac{1}{\sqrt{2\pi}\gamma_i}\exp\big(-\frac{(y_i - \nu_i)^2}{2\gamma_i^2}\big), \\ [2mm]
h(x_1,x_2)=H(e^{x_1},e^{x_2}), \\ [2mm]
u(\tau,x)=e^{-r(T-t)}V(t,e^{x_1},e^{x_2}),
\end{array}\right.
\end{eqnarray}
where $\lambda_i$ is the intensity of normal distributed jumps with mean $\nu_i$ and variance $\gamma_i^2$ (cf. \cite{Rama:2003,Zhou:2009}). We now reformulate (\ref{eq:PIDE}) into an operator version as follows
\begin{eqnarray}\label{PIDE:OperatorForm}
u_\tau = \D[u]+\J[u], \quad (\tau,x) \in [0,T]\times  \mathbb{R}^2
\end{eqnarray}
together with initial condition
\begin{eqnarray}\label{PIDE:OriginalInitialCondition}
u|_{\tau=0}=h(x),\quad x \in \mathbb{R}^2,
\end{eqnarray}
where
\begin{eqnarray*}
\D[u] &=& \nabla\cdot (\kappa \nabla u)+\nabla \cdot (\alpha u) = \sum_{i,j=1}^{2}\kappa_{i,j} \frac{\partial^2 u}{\partial x_i \partial x_j} +\sum_{i=1}^2\alpha_i \frac{\partial u}{\partial x_i}, \\
\J[u] &=& \quad \int_{\mathbb{R}}\left(u(\tau,x+y_1 \e_1) - u(\tau,x)-(e^{y_1}-1) \frac{\partial u}{\partial x_1} (\tau,x)\right) k_1(y_1)dy_1 \\[2mm]
&& + \int_{\mathbb{R}}\left(u(\tau,x+y_2 \e_2)-u(\tau,x)-(e^{y_2}-1) \frac{\partial u}{\partial x_2}(\tau,x)\right)k_2(y_2)dy_2,
\end{eqnarray*}
and $\e_1=(1,0)$, $\e_2=(0,1)$, $\kappa=(\kappa_{i,j})_{2\times2}=\frac{1}{2}\left(\begin{array}{cc}
  \sigma_1^2 & \rho \sigma_1 \sigma_2 \\
  \rho \sigma_1 \sigma_2 & \sigma_2^2 \\
\end{array}%
\right)$, $\alpha=(\alpha_1,\alpha_2)^t=(r-\frac{1}{2}\sigma_1^2,
r-\frac{1}{2}\sigma_2^2)^t$ and $\nabla u=(\frac{\partial u}{\partial x_1},\frac{\partial u}{\partial x_2})^t$.

\subsection{Initial and boundary conditions}
The initial and boundary conditions are desired in order to obtain the solution $u(\tau,x)$  for the pricing equation. The specific payoff structure of the option could provide such information for $u$ when $\tau=0$ or $x_1 \rightarrow \pm \infty$, $x_2 \rightarrow \pm \infty$. From the viewpoint of discounting principle, $u(\tau,x) $ may be set as a reasonable linear function of $e^{x_1}$ and $e^{x_2}$ as $x_1 \rightarrow \pm \infty$ or $x_2 \rightarrow \pm \infty$, where there do not exist jumps at the boundary. So we may consider the initial and boundary
conditions based on a function of the form
\[g(\tau, x) = c_1(\tau)e^{x_1} + c_2(\tau)e^{x_2} + c_3(\tau)>0\]
satisfying equation (\ref{PIDE:OperatorForm}).
Substituting it into (\ref{PIDE:OperatorForm}), we have
\[c_1'e^{x_1} + c_2'e^{x_2} + c_3' = rc_1e^{x_1} + rc_2e^{x_2}\]
for all $x$. Thus $c_1(\tau) = e^{r\tau}c_1(0)$, $c_2(\tau) = e^{r\tau}c_2(0),$ and $c_3$ is constant.

Now we just consider a put option as an example for writing the initial and boundary conditions. The conditions for a call option can be written accordingly. For a European basket put with payoff $\max(K-(w_1S_1+w_2S_2),0)$, we may take $c_1(0) = -w_1$, $c_2(0) = -w_2$ and $c_3(0)= K$.
Then we may write
\begin{eqnarray}\label{cond:boundary_cond_at_+infty}
g(\tau, x) = \left(K-w_1e^{x_1+r\tau} - w_2e^{x_2+r\tau}\right)^+
\end{eqnarray}
and the initial and boundary conditions for pricing equation  (\ref{PIDE:OperatorForm}) may be given as follows
\begin{eqnarray}\label{cond:initial&boundary}
\begin{array}{l}
h(x) = (K-w_1e^{x_1}-w_2e^{x_2})^+,\quad x \in \Omega,\\
g(\tau,x) = (K-w_1e^{x_1+r\tau}- w_2e^{x_2+r\tau})^+,\quad (\tau,x) \in [0,T]\times \partial \Omega,
\end{array}
\end{eqnarray}
where $z^+=\max(0,z)$, $\Omega = \mathbb{R}^2$ and $\partial \Omega=\{(x)\in \mathbb{R}^2|x_1\rightarrow \pm \infty$ or $x_2\rightarrow \pm \infty\}$. From the expression of initial and boundary conditions we have $g(0,x)=h(x)$.

Combined with the initial and boundary conditions (\ref{cond:initial&boundary}), the pricing PIDE (\ref{PIDE:OperatorForm}) for basket put option under exponential \levy model can be written as the following
\begin{eqnarray}\label{eq:PIDE4NonsmoothPayoff}
\left\{\begin{array}{l}
u_{\tau}=\D[u]+\J[u],\\[1mm]
u|_{\tau=0} = h, \quad x \in \Omega,\\[1mm]
u|_{\partial \Omega} = g, \quad \tau\in (0,T].
\end{array}\right.
\end{eqnarray}

The function $g(\tau,x)$ used to define the initial and boundary conditions will be differentiated once in time and twice in space in later sections. For compatibility we assume $h(x)=g(0,x)$. So we need to approximate $g(\tau,x)$ by a smoother function $\tilde{g}(\tau,x)$.
We will then replace the original PIDE by the following
\begin{eqnarray}\label{eq:PIDE4SmoothedPayoff}
\left\{\begin{array}{l}
\tilde{u}_{\tau}=\D[\tilde{u}]+\J[\tilde{u}],\\[1mm]
\tilde{u}|_{\tau=0} = \tilde{h}, \quad x \in \Omega,\\[1mm]
\tilde{u}|_{\partial \Omega}= \tilde{g}, \quad \tau \in [0,T],
\end{array}\right.
\end{eqnarray}
where $\tilde{h}(x)\triangleq \tilde{g}(0,x)$.

To construct a smoother function $\tilde{g}$, we define a curve $C$ and a banded domain $\Omega_{\delta}$ along $C$.
\begin{eqnarray*}
C & = & \left\{x \in \mathbb{R}^2 ~|~ Ke^{-r\tau}-e^{x_1}- e^{x_2} = 0\right \},\\
\Omega_{\delta}& = & \left \{ x \in \Omega ~|~
\mbox{dist} (x, C) < \delta,\; \mbox{~and } \delta > 0 \right \}.
\end{eqnarray*}

For any $x\in \Omega_{\delta}\setminus C$ we can find a point $x^o=(x_1^o,x_2^o) \in C$ such that the segment $\overline{xx^o}$ is normal to $C$ at $x^o$. That is, $x^o$ satisfies $e^{x_1^o+r\tau}+e^{x_2^o+r\tau}=K$ and $(x_1-x^o_1)e^{x^o_1} + (x_2-x^o_2)e^{x^o_2} =0$.
Now we can introduce a smoother function $\tilde{g}$ by redefining $g$ in $\Omega_{\delta}$ along the normal direction $\vec{n}=\left(\frac{
\exp(x_1^o)}{\sqrt{\exp(2x_1^o)+\exp(2x_2^o)}},\frac{
\exp(x_2^o)}{\sqrt{\exp(2x_1^o)+\exp(2x_2^o)}}\right)$ of $C$:
\begin{eqnarray}\label{eq:construction_of_g}
\tilde{g}(\tau,x)=\left \{
\begin{array}{ll}
 p(n), &\quad x \in \Omega_{\delta}  \\
 g(\tau,x),& \quad \mbox{Otherwise},
 \end{array}%
\right.
\end{eqnarray}
where $n=sgn(x_1-x_1^o) \cdot \sqrt{(x_1-x_1^o)^2+(x_2-x_2^o)^2}$ (noting that $n=\pm \delta$ corresponds to points
$x^o\pm \delta \vec{n}$ in original coordinates) and $p(n)$ is a polynomial to be defined as follows.
%
Again we will consider the put option as an example. The call option case can be done similarly. Following the idea for one dimensional case in \cite{Lin:2008}, the polynomial $p(n)$ is constructed along the normal direction $\vec{n}$ satisfying
\begin{eqnarray*}
&p(\delta)=0, p'(\delta)=0, p''(\delta)=0,&\\
&p(-\delta)=g|_{n=-\delta}, p'(-\delta)=\frac{\partial g}{\partial
\vec{n}}|_{n=-\delta}, p''(-\delta)=\frac{\partial^2 g}{\partial^2
\vec{n}}|_{n=-\delta}.&
\end{eqnarray*}
We can explicitly write it as
\[p(n)=(n-\delta)^3 \left(a +b (n+\delta)+c(n+\delta)^2\right ),\]
where
\begin{eqnarray*}
a & = & - \frac{1}{8\delta^3} g|_{n=-\delta}, \\
b & = & - \frac{1}{8\delta^3}\left( \frac{3}{2\delta}g|_{n=-\delta}-\frac{\partial g}{\partial \vec{n}}|_{n=-\delta}\right), \\
c & = & -\frac{1}{16\delta^3}\left( \frac{\partial^2 g}{\partial^2
\vec{n}}|_{n=-\delta}+\frac{3}{\delta}\frac{\partial g}{\partial
\vec{n}}|_{n=-\delta}+ \frac{3}{\delta^2} g|_{n=-\delta} \right).
\end{eqnarray*}
It is easy to verify that $p(n)$ is monotonically decreasing along the normal direction $\vec{n}$ of the curve $C$ and the first order derivatives of $\tilde{g}$ in $\tau$ and the second order derivative of $\tilde{g}$ in $x$ are continuous. Moreover, we have the following proposition about the error between $\tilde{g}$ and $g$.

\begin{proposition} \label{estimate_g} At any fixed $\tau \in [0,T]$,
\begin{eqnarray*}
\left|\tilde{g}-g\right|\leq M \delta^2, \forall x\in \Omega,
\end{eqnarray*}
where $M$ is a generic constant independent of $\delta$. Furthermore,
\[p(n)=\mbox{\textsl{O}}(\delta), p'(n) = \mbox{\textsl{O}}(1), p''(n)=\mbox{\textsl{O}}(\frac{1}{\delta}).\].
\end{proposition}

\begin{figure}[htbp]
\centering
\includegraphics[width=0.7\textwidth]{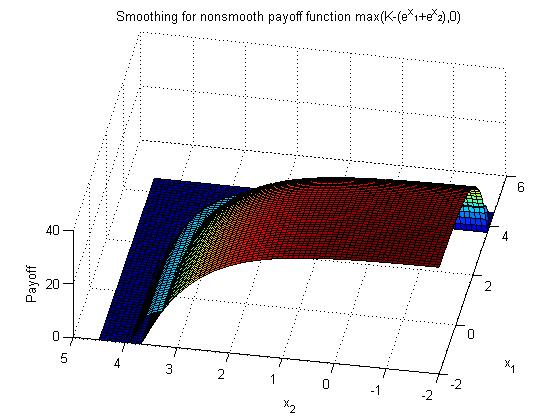}
\caption{Smoothing of the nonsmooth payoff function $max(K-(e^{x_1}+e^{x_2}),0)$ at the initial time.}
\label{Fig:TransformedSmoothvsNonsmooth}
\end{figure}

At time $\tau =0$, we can see the function after the smoothing in Figure \ref{Fig:TransformedSmoothvsNonsmooth}. We can directly apply the similar smoothing technique to the original payoff function $H(S)=\max(K-(S_1+S_2),0)$. Its smoothing version is shown in Figure \ref{Fig:SmoothvsNonsmooth}.
\begin{figure}[htbp]
\centering
\includegraphics[width=0.7\textwidth]{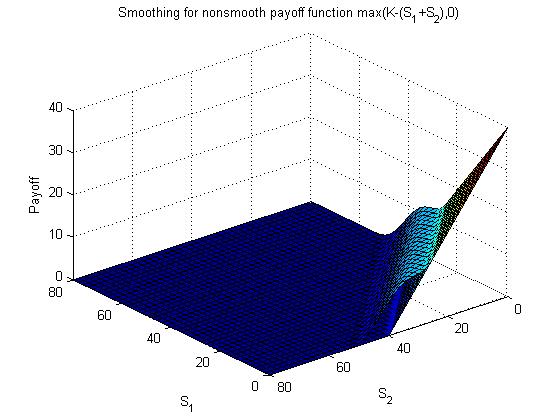}
\caption{Smoothing of the nonsmooth payoff function $\max(K-(S_1+S_2),0)$ in original coordinates $(S_1,S_2)$.}
\label{Fig:SmoothvsNonsmooth}
\end{figure}


Now we estimate the error between the solution (denoting as $u$) of (\ref{eq:PIDE4NonsmoothPayoff}) and the solution (denoting as $\tilde{u}$ of (\ref{eq:PIDE4SmoothedPayoff}), using a maximum principle given in \cite{Maria:1995}.
%
For this purpose we denote our PIDE operator
\[\A[u]=u_{\tau}-\D[u]-\J[u].\]
Then the error $e=u-\tilde{u}$ satisfies $\A[e]=0$, $e|_{\tau=0}=h-\tilde{h}$ and $e|_{\partial \Omega} =g-\tilde{g}$. The operator $\A$ is a special case (corresponding to $a_0=0$ and $f=0$) of the operator considered in \cite{Maria:1995} and all the conditions assumed in \cite{Maria:1995} are satisfied by the exponential \levy model we consider. Simply applying the maximum principle given in Theorem 4.1 of \cite{Maria:1995} for both $e$ and $-e$ (noting $\pm e|_{\tau=0} \leq |h-\tilde{h}|$ and $\pm e|_{\partial \Omega} \leq |g-\tilde{g}|$) and Proposition \ref{estimate_g} we obtain the following error estimate.
\begin{proposition}
For the solution $u$ of (\ref{eq:PIDE4NonsmoothPayoff}) and the solution $\tilde{u}$ of (\ref{eq:PIDE4SmoothedPayoff}) we have
\[|u-\tilde{u}| \leq M \delta^2 \quad \mbox{for all } \tau \in [0, T] \; \mbox{and } x \in \Omega, \]
where $M$ is a generic constant independent of $\delta$.
\end{proposition}

\section{Finite Element Method for PIDE}
In this section, we will give error analysis for the localization of the unbounded domain $\Omega$ to a bounded domain and  for the temporal and spatial discretization of the PIDE under a weighted Sobolev setting.
\subsection{Variational setting}
Since the initial and boundary conditions grows exponentially at the boundary $\partial \Omega$, a weighted Sobolev space $H_{\eta}^{1}(\Omega)$ should be introduced to account for such effects. Given function $\eta(x) = - \left( \eta_1 |x_1| + \eta_2 |x_2| \right)$, we define the weighted Sobolev spaces:
\begin{eqnarray*}
L_{\eta}^2(\Omega) &\triangleq & \left\{ u \in L_{loc}^1(\Omega) \big | u e^{\eta(x)} \in L^2(\Omega) \right\}, \\
H_{\eta}^{1}(\Omega) & \triangleq & \left\{ u \in L_{loc}^{1}(\Omega)\big |u e^{\eta(x)} \in L^2(\Omega), \nabla u e^{\eta(x)}\in \left(L^2(\Omega)\right)^2 \right\}.
\end{eqnarray*}
We note that the functions $g,h,\tilde{g},\tilde{h} \in H_{\eta}^{1}(\Omega)$ for $\eta=(\eta_1,\eta_2)$ satisfying $\eta_i>1, i=1,2$. Later on we will use $\eta$ and $\eta(x)$ interchangeably. For any  vector and matrix $B$, $B^t$ is the transpose of $B$.

By picking up a test function $v\in C_0^{\infty}(\Omega)$ and recalling Green's formula, we consider the variational formulation of (\ref{eq:PIDE4SmoothedPayoff}), which is to

\emph{Find $\tilde{u} \in L^2 \left( [0,T];H_{\eta}^1(\Omega)\right) \bigcap
H^1 \left( [0,T];\left(H_{\eta}^1(\Omega)\right)^*\right)$
such that for any $\tau \in [0,T]$,
\begin{eqnarray}\label{eq:variational_formulation} \left \{
\begin{array}{rcll}
(\frac{\partial \tilde{u}}{\partial \tau},v)_{L_{\eta}^2(\Omega)}+a^{\eta}(\tilde{u},v)&=&\mathcal{A}[\tilde{g}],& \forall ~ v \in H_{\eta}^1(\Omega),\\
\tilde{u}(0,x)&=&0,& \forall ~ x \in \Omega,\\
\tilde{u}(\tau,x) &=& 0, & \forall ~ (\tau,x)\in [0,T]\times
\partial \Omega.
\end{array}
\right.
\end{eqnarray}
}
where for any function space $V$, $V^{*}$ is the dual of $V$ and
\begin{eqnarray}
\nonumber \mathcal{A}[\tilde{g}] &=& -\left(\frac{\partial \tilde{g}}{\partial \tau},v\right)_{H_{\eta}^1(\Omega)} - a^{\eta}(\tilde{g},v),\\
\nonumber \left(\frac{\partial \tilde{u}}{\partial \tau},v\right)_{L_{\eta}^2(\Omega)}&=&
\int_{\Omega}\frac{\partial \tilde{u}}{\partial \tau} v e^{2\eta(x)}dx,\\
\nonumber a^{\eta}(\tilde{u},v)&=&\int_{\Omega}\left( (\nabla \tilde{u})^t \kappa ~ \nabla v -\alpha(x)^t \nabla \tilde{u} ~v \right)e^{2\eta(x)}dx\\
\nonumber && \quad -\int_{\Omega}\J[\tilde{u}]~v e^{2\eta(x)} dx,\\
\nonumber \alpha(x)&=& \alpha + \kappa \nabla \eta,\\
\nonumber \alpha &=&(\alpha_1,\alpha_2)^t=(r-\frac{1}{2}\sigma_1^2, r-\frac{1}{2}\sigma_2^2)^t.
\end{eqnarray}

\begin{remark} If $\eta(x)$ is defined as follows:
\begin{eqnarray}\label{eta}
\eta(x)=\left\{
\begin{array}{ll}
\eta_{1} (|x_1|+ |x_2|)& \textrm{if $x_1<0,x_2<0$,}\\
\eta_{1} |x_1|+ \eta_2 |x_2| & \textrm{if $x_1<0,x_2>0$,}\\
\eta_{2} |x_1|+ \eta_1 |x_2| & \textrm{if $x_1>0,x_2<0$,}\\
\eta_{2} (|x_1|+ |x_2|)& \textrm{if $x_1>0,x_2>0$,}
\end{array}
\right.
\end{eqnarray}
where $\eta_1>1$ and $\eta_2 >1$,
then
\begin{eqnarray}\label{cond:gradient_of_eta}
\nabla \eta =
\left(%
\begin{array}{c}
  \frac{\partial\eta(x)}{\partial x_1} \\
  \frac{\partial\eta(x)}{\partial x_2} \\
 \end{array}%
\right) = \left\{
\begin{array}{ll}
(-\eta_1,-\eta_1)^t & \textrm{if $x_1<0,x_2<0$,}\\
(-\eta_1,\eta_2)^t & \textrm{if $x_1<0,x_2>0$,}\\
(\eta_2,-\eta_1)^t & \textrm{if $x_1>0,x_2<0$,}\\
(\eta_2,\eta_2)^t & \textrm{if $x_1>0,x_2>0$.}
\end{array}
\right.
\end{eqnarray}
Appendix \ref{AppendixA} provides the proof for the following property  $\eta(x)$ satisfies $\eta\in L_{loc}^1(\mathbb{R})$,
$\nabla \eta \in (L^{\infty}(\mathbb{R}))^2$ and
\begin{eqnarray}
\Delta_{\theta}\eta=\eta(x+\theta y)-\eta(x)\leq \eta(y), \quad
\forall x,y \in \mathbb{R}^2,\norm{\theta}\leq 1.
\end{eqnarray}
\end{remark}

\begin{remark}
Note that $a^{\eta}(\tilde{u},v)$ is a nonsymmetric bilinear variational
formulation. when $|\rho| < 1$, $\kappa =\frac{1}{2}\left(%
\begin{array}{cc}
  \sigma_1^2 & \rho \sigma_1 \sigma_2 \\
  \rho \sigma_1 \sigma_2 & \sigma_2^2 \\
\end{array}%
\right)$ is symmetric positive definite, which implies that there
exist two positive numbers $0<\underline{\kappa}\leq
\overline{\kappa}$ such that $\underline{\kappa}|\xi|^2\leq \xi^t
\kappa \xi \leq \overline{\kappa}|\xi|^2$, $~ \forall ~ \xi \in
\mathbb{R}^2$. Also $\alpha(x)$ is uniformly bounded, i.e., given
$\alpha(x)=\left(\alpha_1(x),\alpha_2(x)\right)^t$, there exists
$\overline{\alpha}>0$ such that $|\alpha_i(x)|\leq
\overline{\alpha},\forall~ x\in \Omega,i =1,2.$
\end{remark}

To consider the continuity  and coercivity of the non-symmetric bilinear form $a^{\eta}(\cdot,\cdot)$ in coping with the uniqueness and existence of numerical solution
of variational formulation (\ref{eq:variational_formulation}), we introduce the following \garding inequality. For later convenience, we denote
$\nabla_{\eta} \tilde{u} = \nabla \tilde{u} ~ e^{\eta(x)}$ and $\nabla_{\eta} v =
\nabla v ~ e^{\eta(x)}.$

\begin{proposition}\label{prop:garding_inequality}
Let $\eta \in \mathbb{R}_+^2$ be arbitrarily fixed. If $|\rho|< 1,$
then
\begin{enumerate}
    \item The bilinear form $a^{\eta}(\cdot,\cdot): H_{\eta}^1(\Omega)\times H_{\eta}^1(\Omega)\rightarrow
    \mathbb{R}$ is continuous, i.e., there exists $\overline{c}>0$ such that
    \[|a^{\eta}(u,v)|\leq \overline{c} \norm{u}_{H_{\eta}^1(\Omega)}\norm{v}_{H_{\eta}^1(\Omega)},
\quad \forall~ u,v \in H_{\eta}^1(\Omega);\]
    \item  There exists $\beta>0$ depending on $\eta$ and $\rho$ such that the new bilinear form $a^{\eta}(u,u)+\beta\cdot(u,u)_{L_{\eta}^2(\Omega)}$ is coercive, i.e., there exists $0 \leq \underline{c}\leq \overline{c}$ depending on $\eta$ and $\rho$ such that (\garding inequality)
    \[a^{\eta}(u,u)\geq \underline{c}\norm{u}_{H_{\eta}^1(\Omega)}^2-\beta \norm{u}_{L_{\eta}^2(\Omega)}^2, \quad \forall ~ u \in H_{\eta}^1(\Omega).\]
\end{enumerate}
\end{proposition}
\begin{proof}
From the definition of $a(u,v),$ we know
\begin{eqnarray}\label{eq:cond_for_diff_term}
\nonumber |a^{\eta}(u,v)|& \leq & \int_{\Omega}\left| \left(\nabla u ~ e^{\eta(x)}\right)^t \kappa ~ \left(\nabla v ~ e^{\eta(x)}\right)\right| dx\\
&& + \int_{\Omega}\left| \alpha(x)^t \left(\nabla u ~
e^{\eta(x)}\right) \cdot\left(v~ e^{\eta(x)}\right)\right| dx +
|a_{jump}^{\eta}(u,v)|,
\end{eqnarray}
where $a_{jump}^{\eta}(u,v)=\int_{\Omega} \J[u]~v e^{2\eta(x)}dx.$

For any vectors $\nabla_{\eta} u, \nabla_{\eta} v$ and real number $\lambda$, $\int_{\Omega}(\nabla_{\eta} u + \lambda \nabla_{\eta} v)^t \kappa ~ (\nabla_{\eta} u + \lambda \nabla_{\eta} v)dx > 0$. The first term of RHS in (\ref{eq:cond_for_diff_term}) can be written as follows
\begin{eqnarray*}
\int_{\Omega}\left| \left(\nabla_{\eta} u \right)^t \kappa ~
\nabla_{\eta} v \right| dx &\leq&  \left ( \int_{\Omega}\left|
\left(\nabla_{\eta} u \right)^t \kappa ~ \nabla_{\eta} u
\right|dx\right)^{\frac{1}{2}}  \cdot
\left (\int_{\Omega}\left| \left(\nabla_{\eta} v \right)^t \kappa ~ \nabla_{\eta} v \right| dx \right ) ^{\frac{1}{2}} \\
&\leq & \overline{\kappa} \left ( \int_{\Omega}\left|
\left(\nabla_{\eta} u \right)^t ~ \nabla_{\eta} u
\right|dx\right)^{\frac{1}{2}}  \cdot
\left (\int_{\Omega}\left| \left(\nabla_{\eta} v \right)^t ~ \nabla_{\eta} v \right| dx \right ) ^{\frac{1}{2}}\\
&=&  \overline{\kappa} ~ \norm{\nabla u}_{L_{\eta}^2(\Omega)}
\norm{\nabla v}_{L_{\eta}^2(\Omega)} \leq \overline{\kappa} ~
\norm{u}_{H_{\eta}^1(\Omega)} \norm{v}_{H_{\eta}^1(\Omega)}.
\end{eqnarray*}
Similarly, rewriting the second term of RHS in
(\ref{eq:cond_for_diff_term}) as following:
\begin{eqnarray}\label{ineqality:continuity_for_convection_term}
\nonumber \int_{\Omega}\left| \alpha(x)^t \nabla_{\eta}u  \cdot
\left(v~ e^{\eta(x)}\right)\right| dx &\leq & \overline{\alpha}
\left(\norm{\frac{\partial u}{\partial x_1}}_{L_{\eta}^2(\Omega)} +
\norm{\frac{\partial u}{\partial x_2}}_{L_{\eta}^2(\Omega)}\right)
\norm{v}_{L_{\eta}^2(\Omega)}\\
&\leq &  \overline{\alpha} \norm{\nabla u}_{L_{\eta}^2(\Omega)}
\norm{v}_{L_{\eta}^2(\Omega)} \leq \overline{\alpha} ~
\norm{u}_{H_{\eta}^1(\Omega)} \norm{v}_{H_{\eta}^1(\Omega)}.
\end{eqnarray}

To consider the last term of RHS in (\ref{eq:cond_for_diff_term}) we need to add $\exp(-\eta_1|x_1+\theta_1 y|-\eta_2 |x_2|)$ and
$\exp(-\eta_1|x_1|-\eta_2 |x_2+\theta_2 y|)$ to $\frac{\partial
u}{\partial x_1}(\tau,x_1+\theta_1y,x_2)$ and $\frac{\partial
u}{\partial x_2}(\tau,x+\theta_2y)$, respectively. Thus there
will have extra terms $\exp\left(\eta_i(|x_i+\theta_i
y|-|x_i|)\right)$ satisfying: $\exp\left(\eta_i(|x_i+\theta_i
y|-|x_i|)\right)\leq \exp (\eta_i |y| )$ for $0 \leq \theta_i\leq
1,i=1,2.$ Therefore the last term of RHS in (\ref{eq:cond_for_diff_term}) can also be rewritten as
\begin{eqnarray*}
|a_{jump}^{\eta}(u,v)| &\leq& \int_{\Omega}\int_{\mathbb{R}}\int_0^1
\big|\frac{\partial u}{\partial x_1}(\tau,x_1+\theta_1y,x_2)
e^{\eta(x)}
ve^{\eta(x)}\big | |y|k_1(y) d\theta_1  dy dx \\
&&+\int_{\Omega}\int_{\mathbb{R}}\big|\frac{\partial u}{\partial x_1}(\tau,x) e^{\eta(x)} v e^{\eta(x)} \big|\cdot|e^y-1|k_1(y)dydx\\
&&+\int_{\Omega}\int_{\mathbb{R}}\int_0^1
\big|\frac{\partial u}{\partial x_2}(\tau,x+\theta_2y)e^{\eta(x)} v e^{\eta(x)}\big| |y|k_2(y) d\theta_2 dy dx \\
&&+\int_{\Omega}\int_{\mathbb{R}}\big|\frac{\partial u}{\partial x_2}(\tau,x) e^{\eta(x)} v e^{\eta(x)} \big|\cdot|e^y-1|k_2(y)dydx\\
&\leq& c_1 \norm{\frac{\partial u}{\partial
x_1}}_{L_{\eta}^2(\Omega)} \norm{v}_{L_{\eta}^2(\Omega)}+ c_2
\norm{\frac{\partial u}{\partial x_2}}_{L_{\eta}^2(\Omega)}
\norm{v}_{L_{\eta}^2(\Omega)}\\
&\leq&\max(c_1,c_2)\cdot \norm{\nabla
u}_{L_{\eta}^2(\Omega)}\norm{v}_{L_{\eta}^2(\Omega)} \\
&\leq&\max(c_1,c_2) \norm{u}_{H_{\eta}^1(\Omega)}\norm{v}_{H_{\eta}^1(\Omega)},
\end{eqnarray*}
where
\[c_i=\int_{\mathbb{R}}\left(e^{\eta_i |y|}|y|+|e^y-1|\right)k_i(y)dy,i=1,2.\]
For Gaussian $k_i(y)$, $c_1,c_2$ are positive and finite. Thus we have the continuity condition of $a^{\eta}(u,v)$
\[|a^{\eta}(u,v)|\leq \overline{c} \norm{ u}_{H_{\eta}^1(\Omega)}\norm{v}_{H_{\eta}^1(\Omega)},\]
where
$\overline{c}=\overline{\kappa}+\overline{\alpha}+\max(c_1,c_2).$

Now the coercivity of $a^{\eta}(u,u)$ remains to be proved. Since
\begin{eqnarray*}
a^{\eta}(u,u)&=&\int_{\Omega} \left(\nabla u ~ e^{\eta(x)}\right)^t \kappa ~ \left(\nabla u ~ e^{\eta(x)}\right) dx\\
&& - \int_{\Omega} \alpha(x)^t \left(\nabla u ~ e^{\eta(x)}\right)
\cdot\left(u~ e^{\eta(x)}\right) dx - a_{jump}^{\eta}(u,u),
\end{eqnarray*}
then
\begin{eqnarray}\label{eq:cond_for_diff_term1}
\nonumber a^{\eta}(u,u)&\geq&\int_{\Omega}\left| \left(\nabla u ~ e^{\eta(x)}\right)^t \kappa ~ \left(\nabla u ~ e^{\eta(x)}\right)\right| dx\\
&& - \int_{\Omega}\left| \alpha(x)^t \left(\nabla u ~
e^{\eta(x)}\right) \cdot\left(u~ e^{\eta(x)}\right)\right| dx -
|a_{jump}^{\eta}(u,u)|.
\end{eqnarray}

Since $\kappa$ is positive definite when $|\rho|<1$, the first term of (\ref{eq:cond_for_diff_term1}) can be simplified as
\begin{eqnarray}\label{coercivity-quartic}
\nonumber \int_{\Omega}\left| \left(\nabla_{\eta} u\right)^t \kappa ~ \nabla_{\eta} u \right| dx & \geq & \underline{\kappa} \int_{\Omega}\left| \left(\nabla_{\eta} u\right)^t \nabla_{\eta} u \right| dx \\
&=& \underline{\kappa}\left(\norm{\frac{\partial u}{\partial x_1}}_{L_{\eta}^2(\Omega)}^2+\norm{\frac{\partial u}{\partial x_2}}_{L_{\eta}^2(\Omega)}^2 \right).
\end{eqnarray}
By (\ref{ineqality:continuity_for_convection_term}), the second term of RHS in (\ref{eq:cond_for_diff_term1}) can also be rewritten as
\begin{eqnarray}
 \nonumber -\int_{\Omega}\left| \alpha(x)^t \nabla_{\eta}u \cdot\left(u~ e^{\eta(x)}\right)\right| dx \\
\geq -\overline{\alpha} \left(\norm{\frac{\partial u}{\partial x_1}}_{L_{\eta}^2(\Omega)} +
\norm{\frac{\partial u}{\partial x_2}}_{L_{\eta}^2(\Omega)}\right) \norm{v}_{L_{\eta}^2(\Omega)}.
\end{eqnarray}
Finally, we analyze the last integral term of RHS of (\ref{eq:cond_for_diff_term1}) in the following form
\begin{eqnarray}\label{coercivity-integral}
\nonumber -|a_{jump}^{\eta}(u,u)| &\geq&
\int_{\Omega}\int_{\mathbb{R}}\int_0^1
\big|\frac{\partial u}{\partial x_1}(\tau,x_1+\theta_1y,x_2) e^{\eta(x)} u e^{\eta(x)}\big | |y|k_1(y) d\theta_1  dy dx \\
\nonumber &&-\int_{\Omega}\int_{\mathbb{R}}\big|\frac{\partial u}{\partial x_1}(\tau,x) e^{\eta(x)} u e^{\eta(x)} \big|\cdot|e^y-1|k_1(y)dydx\\
\nonumber &&-\int_{\Omega}\int_{\mathbb{R}}\int_0^1
\big|\frac{\partial u}{\partial x_2}(\tau,x+\theta_2y)e^{\eta(x)} u e^{\eta(x)}\big| |y|k_2(y) d\theta_2 dy dx \\
\nonumber &&-\int_{\Omega}\int_{\mathbb{R}}\big|\frac{\partial u}{\partial x_2}(\tau,x) e^{\eta(x)} u e^{\eta(x)} \big|\cdot|e^y-1|k_2(y)dydx\\
 &\geq& -c_1 \norm{\frac{\partial u}{\partial x_1}}_{L_{\eta}^2(\Omega)} \norm{u}_{L_{\eta}^2(\Omega)}- c_2  \norm{\frac{\partial u}{\partial x_2}}_{L_{\eta}^2(\Omega)}  \norm{u}_{L_{\eta}^2(\Omega)}.
\end{eqnarray}
Then from (\ref{eq:cond_for_diff_term1}-\ref{coercivity-integral}), we obtain
\begin{eqnarray*}
a^{\eta}(u,u)&\geq&\underline{\kappa}\left(\norm{\frac{\partial
u}{\partial x_1}}_{L_{\eta}^2(\Omega)}^2+\norm{\frac{\partial
u}{\partial x_2}}_{L_{\eta}^2(\Omega)}^2 \right)
-\sum_{i=1}^2(\overline{\alpha}+c_i) \norm{\frac{\partial u}{\partial
x_i}}_{L_{\eta}^2(\Omega)}\norm{u}_{L_{\eta}^2(\Omega)}.
\end{eqnarray*}
For any $\varepsilon >0$, we have
\[\norm{\frac{\partial u}{\partial
x_i}}_{L_{\eta}^2(\Omega)}\norm{u}_{L_{\eta}^2(\Omega)}\leq
\varepsilon \norm{\frac{\partial u}{\partial
x_i}}_{L_{\eta}^2(\Omega)}^2+\frac{1}{4\varepsilon}\norm{u}_{L_{\eta}^2(\Omega)}^2.
\]
By selecting $\varepsilon$ satisfying $\underline{\kappa}-\varepsilon(\overline{\alpha}+c_i)>0$ and
defining
\begin{eqnarray*}
\underline{c}&=&\underline{\kappa}-\varepsilon \overline{\alpha}- \varepsilon \max(c_1,c_2)>0,\\
\beta &=& \underline{c}+\frac{1}{4\varepsilon}(2\overline{\alpha}+c_1+c_2)>0,
\end{eqnarray*}
we have the coercivity property:
\[a^{\eta}(u,u)\geq \underline{c}\norm{u}_{H_{\eta}^1(\Omega)}^2-\beta\norm{u}_{L_{\eta}^2(\Omega)}^2.\]
\end{proof}

\begin{proposition}\label{prop:estimationfortilde_u}The variational formulation (\ref{eq:variational_formulation}) has a unique solution $\tilde{u}$ and the following estimate, for all $\tau \in [0,T]$
\begin{eqnarray*}\label{eq:priori-estimation}
\begin{array}{l}
\dsp e^{-2\beta \tau}\norm{\tilde{u}}_{L_{\eta}^2(\Omega)}^2+\underline{c} \int_0^{\tau}e^{-2\beta s}\norm{\tilde{u}}_{H_{\eta}^1(\Omega)}^2ds \\[2mm]
\dsp \qquad \qquad \leq \frac{1+\overline{c}}{\underline{c}}\int_{0}^{\tau}e^{-2\beta s}\left(\norm{\frac{\partial \tilde{g}}{\partial
s}}_{L_{\eta}^2(\Omega)}^2+\overline{c}\norm{\tilde{g}}_{H_{\eta}^1(\Omega)}^2\right)ds,
\end{array}
\end{eqnarray*}
where $\underline{c}$ and $\overline{c}$ are the constants in Proposition \ref{prop:garding_inequality}.
\end{proposition}

\begin{proof} Since $\tilde{h}(x)\in L_{\eta}^2(\Omega),$ the existence and uniqueness of variational formulation (\ref{eq:variational_formulation}) can be obtained (see \cite{Lions:1972}).

Taking $v(s,x)=\tilde{u}(s,x)e^{-2 \beta s}$ in (\ref{eq:variational_formulation}) and integrating in time $s$ between $0$ and $\tau,$ we derive that:
\begin{eqnarray*}
\frac{1}{2}~e^{-2\beta \tau}\norm{\tilde{u}}_{L_{\eta}^2(\Omega)}^2
&+&\int_0^{\tau} e^{-2\beta s}\left( a^{\eta}(\tilde{u},\tilde{u}) + \beta\cdot(\tilde{u},\tilde{u})_{L_{\eta}^2(\Omega)} \right)ds \\
&=&-\int_0^{\tau}e^{-2\beta s}\left(\frac{\partial \tilde{g}}{\partial s},\tilde{u}\right)_{L_{\eta}^2(\Omega)}ds-
\int_{0}^{\tau}e^{-2\beta s}a^{\eta}(\tilde{g},\tilde{u})ds.
\end{eqnarray*}
Using \garding inequality in Proposition \ref{prop:garding_inequality}, we have
\begin{eqnarray*}
\frac{1}{2}~e^{-2\beta \tau}\norm{\tilde{u}}_{L_{\eta}^2(\Omega)}^2&+&\underline{c}\int_0^{\tau}e^{-2\beta
s}\norm{\tilde{u}}_{H_{\eta}^1(\Omega)}^2ds\\
&\leq& -\int_0^{\tau}e^{-2\beta s}\left(\frac{\partial \tilde{g}}{\partial
s},\tilde{u}\right)_{L_{\eta}^2(\Omega)}ds-\int_{0}^{\tau}e^{-2\beta s}a^{\eta}(\tilde{g},\tilde{u})ds\\
&\leq& \frac{1}{2\epsilon}\int_{0}^{\tau}e^{-2\beta s}\left(\norm{\frac{\partial \tilde{g}}{\partial s}}_{L_{\eta}^2(\Omega)}^2+\overline{c}\norm{\tilde{g}}_{H_{\eta}^1(\Omega)}^2\right)ds\\
&&+\frac{1}{2}\epsilon \int_{0}^{\tau}e^{-2\beta s} \left(\norm{\tilde{u}}_{L_{\eta}^2(\Omega)}^2 + \overline{c}\norm{\tilde{u}}_{H_{\eta}^1(\Omega)}^2\right)ds.
\end{eqnarray*}
Thus we obtain the following estimation
\begin{eqnarray*}
e^{-2\beta
\tau}\norm{\tilde{u}}_{L_{\eta}^2(\Omega)}^2&+&\left(2\underline{c}-\epsilon(1+\overline{c})\right)
\int_0^{\tau}e^{-2\beta s}\norm{\tilde{u}}_{H_{\eta}^1(\Omega)}^2ds\\
&\leq& \frac{1}{\epsilon}\int_{0}^{\tau}e^{-2\beta
s}\left(\norm{\frac{\partial \tilde{g}}{\partial
s}}_{L_{\eta}^2(\Omega)}^2+\overline{c}\norm{\tilde{g}}_{H_{\eta}^1(\Omega)}^2\right)ds.
\end{eqnarray*}
Let $\epsilon=\frac{\underline{c}}{1+\overline{c}}>0$, a priori estimation (\ref{eq:priori-estimation}) is obtained.
\end{proof}

\subsection{Localization to a bounded domain and its error estimate}
Since it is not feasible to derive a numerical solution in an unbounded domain $\Omega$ for \emph{variational formulation} (\ref{eq:variational_formulation}) corresponding to the pricing equation (\ref{eq:PIDE4NonsmoothPayoff}), we need to truncate $\Omega$ to a bounded computational domain $\Omega_{M}=[-M,M]\times [-M,M]$. The domain truncation technique is not necessary for some exotic options such as double barrier option, which directly results in PIDE on a naturally bounded domain.

Instead of solving (\ref{eq:PIDE4SmoothedPayoff}) with smoothed boundary condition $\tilde{g}$ in an unbounded domain $\Omega\times [0,T]$, we will solve the truncated problem on $\Omega_M\times [0,T]:$
\begin{eqnarray}\label{eq:truncated_call}
\left\{\begin{array}{l}
(\tilde{u}_M)_{\tau}=\nabla\cdot (\kappa \nabla \tilde{u}_M)+ \nabla \cdot(\alpha \tilde{u}_M)+\J[\tilde{u}_M],\\[1mm]
\tilde{u}_M|_{\tau=0} = \tilde{h}(x), \quad x \in \Omega_M,\\[1mm]
\tilde{u}_M|_{\partial \Omega} = \tilde{g}(\tau,x), \quad (\tau,x)\in
(0,T] \times \partial \Omega_M.
\end{array}\right.
\end{eqnarray}
The variational formulation of above truncated PIDE (\ref{eq:truncated_call}) is to

\emph{Find $\tilde{u}_M \in L^2 \left( [0,T];H_{\eta}^1(\Omega_M)\right) \bigcap H^1 \left( [0,T]; \left(H_{\eta}^1(\Omega_M)\right)^* \right)$ such that for any $\tau \in [0,T]$,
\begin{eqnarray}\label{eq:truncated_variational_formulation}
\qquad \left\{
\begin{array}{rcll}
(\frac{\partial \tilde{u}_M}{\partial \tau},v)_{L_{\eta}^2(\Omega_M)}+a_M^{\eta}(\tilde{u}_M,v)&=&\mathcal{A}[\tilde{g}],& \forall ~ v \in H_{\eta}^1(\Omega_M)\\
\tilde{u}_M(0,x)&=&0,& \forall ~ x \in \Omega_M\\
\tilde{u}_M(\tau,x) &=& 0, & \forall ~ (\tau,x)\in [0,T] \times \partial \Omega_M.
\end{array}
\right.
\end{eqnarray}
}
where $a_M^{\eta}(\tilde{u}_M,v)$ is the restriction of $a^{\eta}(\tilde{u}_M,v)$ to $\Omega_M$, i.e.,
\[a_M^{\eta}(u,v) = \int_{\Omega_M}\left( (\nabla u)^t \kappa ~ \nabla v -\alpha(x)^t\nabla u ~v \right) e^{2\eta(x)}dx - \int_{\Omega_M}\J[u]~v e^{2\eta(x)} dx.
\]
For convenience we still denote by $\tilde{u}_M$ its extension by zero to all of $\Omega$. The restriction of approximating solution $\tilde{u}$ of
(\ref{eq:variational_formulation}) in $\Omega$ to $\Omega_M$
introduces a localization error $e_M=\tilde{u}_M-\tilde{u}$ which we now estimate. we have the following localization error estimation:
\begin{proposition} Suppose $\Omega_{M/2}\triangleq \{x \in \mathbb{R}^2 \big | \norm{x}\leq M/2\}$. Then there exists positive constant $C$ only depending on $T$ and $\gamma > 0$, which is independent of $M$ such that the localization error $e_M(\tau)$ satisfies:
\begin{eqnarray}\label{eq:localization_err_estimation_inequalty}
\norm{e_M(\tau,\cdot)}_{L^2(\Omega_{M/2})}^2+\int_0^{\tau}
\norm{e_M(s,\cdot)}_{H^1(\Omega_{M/2})}^2ds \leq C e^{-\gamma M}.
\end{eqnarray}
\end{proposition}
\begin{proof}
The a priori estimate in Proposition (\ref{prop:estimationfortilde_u}) implies that $\norm{\tilde{u}}_{L_{\eta}^2(\Omega)}$ is bounded by some constant $C$ independent of $M$. Likewise,
\begin{eqnarray}\label{estimate:tildeu_M}
\nonumber e^{-2\beta
\tau}\norm{\tilde{u}_M}_{L_{\eta}^2(\Omega)}^2&+&\underline{c}
\int_0^{\tau}e^{-2\beta
s}\norm{\tilde{u}_M}_{H_{\eta}^1(\Omega)}^2ds\\
&\leq& \frac{1+\overline{c}}{\underline{c}}\int_{0}^{\tau}e^{-2\beta
s}\left(\norm{\frac{\partial \tilde{g}}{\partial
s}}_{L_{\eta}^2(\Omega)}^2+\overline{c}\norm{\tilde{g}}_{H_{\eta}^1(\Omega)}^2\right)ds.
\end{eqnarray}
Thus
$\norm{e_M}_{L_{\eta}^2(\Omega)},\norm{e_M}_{H_{\eta}^1(\Omega)}$
are bounded by $C$.

For any $\tau \in [0,T],v \in H_0^1(\Omega_M)$, the localization error $e_M=\tilde{u}_M(\tau,x)-\tilde{u}(\tau,x)$ satisfies:
\begin{eqnarray}\label{eq:localization_err_estimation}
\left(\frac{d }{d \tau}e_M(\tau),v \right)_{L^2(\mathbb{R}^2)}+ a \left(e_M(\tau),v \right)=0,
\end{eqnarray}
where $a(\cdot,\cdot)\triangleq a^0(\cdot,\cdot)=a^{\eta}(\cdot,\cdot)|_{\eta=0}$.

Define a cut-off function $\phi$ with the following properties: $\phi \in \mathcal{C}_{0}^{\infty}(\Omega_M),\phi=1$ on $\Omega_{M/2}$ and $\norm{\nabla \phi}_{L^{\infty}(\Omega_M)}\leq C$, where $C$ is independent of $M$. Inserting $v=\phi^2(x)e_M(\tau,x) \in
H_{\eta}^1(\Omega_M) $ in (\ref{eq:localization_err_estimation}), we
have
\begin{eqnarray}\label{eq:localization_error_estimation equation}
\frac{1}{2}\frac{d}{d \tau}\norm{\phi e_M(\tau)}_{L^2(\Omega_M)}^2+a\big(\phi e_M(\tau),\phi
e_M(\tau)\big)=\rho_M(\tau),
\end{eqnarray}
where the residual $\rho_M(\tau)=a_M\big(\phi e_M(\tau),\phi
e_M(\tau)\big)-a\big(e_M(\tau),\phi^2 e_M(\tau)\big).$  The residual $\rho_M(\tau)$ can be rewritten as follows
\begin{eqnarray}\label{eq:localization_error_estimation2}
\rho_M(\tau)&=&\sum_{i=1}^2\frac{\sigma_i^2}{2}\int_{\Omega}|\frac{\partial
\phi}{\partial x_i}|^2|e_M|^2dx_1dx_2 \\
\nonumber && -\sum_{i=1}^2 \alpha_i
\int_{\Omega}\frac{\partial \phi}{\partial
x_i}\phi|e_M|^2dx_1dx_2+\overline{\rho}_M(\tau),
\end{eqnarray}
where $\overline{\rho}_M(\tau)=-\int_{\Omega}\J[\phi
e_M]\phi e_M dx+ \int_{\Omega}\J[e_M]\phi^2 e_M dx.$

The weighted function $\eta(x)$ defined in (\ref{eta}) satisfies $\eta \in \L_{loc}^1(\mathbb{R})$, $\nabla \eta \in \L^{\infty}(\mathbb{R})$ and satisfies
\begin{eqnarray}
\Delta_{\theta}\eta=\eta(x+\theta y)-\eta(x)\leq \eta(y), \quad
\forall x,y \in \mathbb{R}^2,\norm{\theta}\leq 1,
\end{eqnarray}
The first two integral terms in the expression (\ref{eq:localization_error_estimation2}) can be estimated by
\begin{eqnarray}
\nonumber
\left|\sum_{i=1}^2\frac{\sigma_i^2}{2}\int_{\Omega}|\frac{\partial
\phi}{\partial x_i}|^2|e_M|^2dx_1dx_2 -\sum_{i=1}^2 \alpha_i
\int_{\Omega}\frac{\partial \phi}{\partial
x_i}\phi|e_M|^2dx_1dx_2\right|\\
\nonumber \leq C\int_{\Omega_M \backslash \Omega_{M/2}} |e_M|^2
e^{\eta(x)}e^{-\eta(x)}dx \leq C e^{-\gamma M}\norm{e_M}_{L_{\eta}^2(\Omega)}^2,
\end{eqnarray}
for positive constants $C$ and $\gamma$, where $\gamma$ independent of $M$.

Now $\overline{\rho}_M(\tau)$ is still left to estimate. Denote $K_i(z)$ the first anti-derivative of \levy measure $k_i$, i.e.,
\begin{eqnarray}
K_i(z)= \left \{
\begin{array}{c}
\int_z^{\infty} k_i(y)dy,\quad \mbox{if } z>0,\\[3mm]
-\int_{-\infty}^z k_i(y)dy,\quad \mbox{if } z<0,
\end{array}
\right.
\end{eqnarray}
then
\begin{eqnarray*}
\J[u]& = & \int_{\mathbb{R}}\left(\frac{\partial u}{\partial x_1}(x+z \e_1)-\frac{\partial u}{\partial x_1}(x)\right)K_1(z)dz\\
  && +\int_{\mathbb{R}}\left(\frac{\partial u}{\partial x_2}(x+z \e_2)-\frac{\partial u}{\partial x_2}(x)\right)K_2(z)dz + \chi \nabla u,
\end{eqnarray*}
where $\e_1=(1,0),\e_2=(0,1)$, $\chi=(\chi_1,\chi_2)$ and
\[\chi_i = \int_{\mathbb{R}}[y-(e^y-1)]k_i(y)dy.\]

We could decompose $\overline{\rho}_M(\tau)$ as follows
\[\overline{\rho}_M(\tau)=I_{11}+I_{12}+I_{21}+I_{22}+I_{3},\]
where
\begin{eqnarray*}
I_{11}&=&-\int_{\mathbb{R}}\int_{\Omega}\Big(\frac{\partial
e_M}{\partial
x_1}(\tau,x+z\e_1)\phi(x+z\e_1)+e_M(\tau,x+z\e_1)\frac{\partial
\phi}{\partial x_1}(x+z\e_1)\\
&&\qquad \qquad -\frac{\partial e_M}{\partial
x_1}(\tau,x)\phi(x)-e_M(\tau,x)\frac{\partial
\phi}{\partial x_1}(x)
\Big)\phi e_M K_1(z)dxdz,\\
I_{12}&=&-\int_{\mathbb{R}}\int_{\Omega}\Big(\frac{\partial
e_M}{\partial
x_2}(\tau,x+z\e_2)\phi(x+z\e_2)+e_M(\tau,x+z\e_2)\frac{\partial
\phi}{\partial x_2}(x+z\e_1)\\
&&\qquad \qquad-\frac{\partial e_M}{\partial
x_2}(\tau,x)\phi(x)-e_M(\tau,x)\frac{\partial
\phi}{\partial x_2}(x)
\Big)\phi e_M K_2(z)dxdz,\\
I_{21}&=&\int_{\mathbb{R}}\int_{\Omega}\big(\frac{\partial e_M}{\partial x_1}(\tau,x+z\e_1)-\frac{\partial e_M}{\partial x_1}(\tau,x) \big)\phi^2 e_M K_1(z)dxdz,\\
I_{22}&=&\int_{\mathbb{R}}\int_{\Omega}\big(\frac{\partial e_M}{\partial x_2}(\tau,x+z\e_2)-\frac{\partial e_M}{\partial x_2}(\tau,x) \big)\phi^2 e_M K_2(z)dxdz,\\
I_3 &=& -\int_{\Omega} \chi \nabla (\phi e_M) \phi e_M dx+\int_{\Omega} \chi \nabla ( e_M) \phi^2 e_M dx.
\end{eqnarray*}
Before we are going to estimate $\overline{\rho}_M(\tau)$ respectively, we need to rearrange $I_{11},I_{12},I_{21},I_{22}$ as follows
\[\overline{\rho}_M(\tau)=\overline{I}_{11}+\overline{I}_{12}+\overline{I}_{21}+\overline{I}_{22}+I_3,\]
where
\begin{eqnarray*}
\overline{I}_{11}&=&-\int_{\mathbb{R}}\int_{\Omega}\frac{\partial
e_M}{\partial
x_1}(\tau,x+z\e_1)(\phi(x+z\e_1)-\phi(x))e_M \phi K_1(z)dxdz,\\
\overline{I}_{12}&=& - \int_{\mathbb{R}}\int_{\Omega}\frac{\partial
e_M}{\partial x_2}(\tau,x+z\e_2)(\phi(x+z\e_2) - \phi(x))e_M \phi K_2(z)dxdz,\\
\overline{I}_{21}&=&-\int_{\mathbb{R}}\int_{\Omega}\Big(e_M(\tau,x+z\e_1) - e_M(\tau,x)\Big)\frac{\partial \phi}{\partial x_1}(x)\phi e_M K_1(z)dxdz \\
&&-\int_{\mathbb{R}}\int_{\Omega} e_M(\tau,x+z\e_1)\Big(\frac{\partial \phi}{\partial x_1}(x+z\e_1) - \frac{\partial \phi}{\partial x_1}(x)\Big)\phi e_M K_1(z) dxdz,\\
\overline{I}_{22}&=&-\int_{\mathbb{R}}\int_{\Omega}\Big(e_M(\tau,x+z\e_2) - e_M(\tau,x)\Big)\frac{\partial \phi}{\partial x_2}(x)\phi e_M K_2(z)dxdz\\
&&-\int_{\mathbb{R}}\int_{\Omega} e_M(\tau,x+z\e_2)\Big(\frac{\partial \phi}{\partial x_2}(x+z\e_2)-\frac{\partial \phi}{\partial x_2}(x)\Big) \phi e_M K_2(z)dxdz,\\
I_3 &=& -\int_{\Omega}  e_M^2 \phi \cdot \chi \nabla \phi  dx.
\end{eqnarray*}

We denote $D_1=\{x \in \Omega|\sqrt{(x_1+z)^2+x_2^2}\geq M/2\}$ and $D_2=\{x \in \Omega|\sqrt{x_1^2+x_2^2}\geq M/2\}$ and observe that the integrand in $I_3$ is zero if $x \notin D_2$. Thus
\begin{eqnarray}\label{err-estimate_I3}
\nonumber|I_3|    &=& \left |-\int_{D_2}  e_M^2 \phi \cdot \chi \nabla \phi dx \right |\\
\nonumber &=&\int_{D_2}  |e_M e^{\eta(x)}|^2 \cdot | \phi| \cdot | \chi \nabla \phi| \cdot e^{-2\eta(x)}dx\\
&\leq & C e^{-\gamma M} \norm{e_M }_{L_{\eta}^2(\Omega_M)}.
\end{eqnarray}
Noting that the integrand in $\overline{I}_{11}$ is zero if $x
\notin D_1 \bigcup D_2$ we have
\begin{eqnarray*}
\overline{I}_{11}=-\int_{\mathbb{R}}\int_{D_1 \bigcup
D_2}\frac{\partial e_M}{\partial
x_1}(\tau,x+z \e_1)(\phi(x+z \e_1)-\phi(x))e_M \phi
K_1(z)dxdz.
\end{eqnarray*}
It implies that
\begin{eqnarray}\label{err-estimate_I11}
\nonumber|\overline{I}_{11}|&\leq& \int_{\mathbb{R}}|z|K_1(z)dz\int_{D_1}\left|\frac{\partial e_M}{\partial x_1}(\tau,x+z \e_1) e^{\eta(x+z \e_1)}\right|\cdot|e_M| \cdot e^{-\eta(x+z \e_1)} \phi dx\\
&& + \int_{\mathbb{R}}|z|K_1(z)dz\int_{D_2}\left |\frac{\partial e_M}{\partial x_1}(\tau,x+z \e_1)\right| \cdot \left|e_M e^{\eta(x)}\right|\cdot e^{-\eta(x)}dx\\
\nonumber &\leq& C e^{-\gamma M} \left(\norm{\frac{\partial e_M}{\partial
x_1}e^{\eta(x)}}_{L^2(\Omega)}\cdot \norm{e_M}_{L^2(\Omega)}+\norm{\frac{\partial e_M}{\partial x_1}}_{L^2(\Omega)} \cdot \norm{e_M e^{\eta(x)}}_{L^2(\Omega)}\right).
\end{eqnarray}
Similarly, the above inequality holds true for $\overline{I}_{12},$
i.e.,
\begin{eqnarray} \label{err-estimate_I12}
\qquad \qquad |\overline{I}_{12}| \leq C e^{-\gamma M} \left(\norm{\frac{\partial e_M}{\partial x_2}}_{L_{\eta}^2(\Omega)}\cdot \norm{e_M}_{L^2(\Omega)}+\norm{\frac{\partial e_M}{\partial x_2}}_{L^2(\Omega)} \cdot \norm{e_M}_{L_{\eta}^2(\Omega)}\right).
\end{eqnarray}
Since two integrand functions in $\overline{I}_{21}$ are zero if $x \in \Omega_{M/2}$, $\overline{I}_{21}$ can be estimated as follows
\begin{eqnarray} \label{err-estimate_I21}
\nonumber |\overline{I}_{21}| \leq \int_{\mathbb{R}}K_1(z)dz\int_{D_2} \left|e_M(\tau,x+z \e_1)-e_M(\tau,x)\right|\cdot\left|\frac{\partial \phi}{\partial x_1}\phi\right| \cdot\left|e_Me^{\eta(x)}\right| \cdot e^{-\eta(x)} dx\\
\nonumber +\int_{\mathbb{R}}K_1(z)dz\int_{D_2}| e_M(\tau,x+ze_1)|\left|\left(\frac{\partial \phi}{\partial x_1}(x+z \e_1) - \frac{\partial \phi}{\partial x_1}(x)\right)\phi\right| \left|e_M e^{\eta(x)}\right|\cdot e^{-\eta(x)} dx,\\
\leq C e^{-\gamma M}\left(\norm{e_M}_{H^1(\Omega_M)} \cdot \norm{e_M}_{L_{\eta}^2(\Omega_M)}+ \norm{e_M}_{L^2(\Omega_M)}\cdot \norm{e_M }_{L_{\eta}^2(\Omega_M)} \right).
\end{eqnarray}
Similarly,
\begin{eqnarray}\label{err-estimate_I22}
\qquad \qquad |\overline{I}_{22}| &\leq& C e^{-\gamma M}\left(\norm{e_M}_{H^1(\Omega_M)}
\cdot \norm{e_M}_{L_{\eta}^2(\Omega_M)}+ \norm{e_M}_{L^2(\Omega_M)}\cdot \norm{e_M}_{L_{\eta}^2(\Omega_M)} \right).
\end{eqnarray}
Integrating (\ref{eq:localization_error_estimation equation}) from $0$ to $\tau$ and combining it with the estimations: (\ref{err-estimate_I3}), (\ref{err-estimate_I11}), (\ref{err-estimate_I12}), (\ref{err-estimate_I21}), (\ref{err-estimate_I22}) and using the a priori estimate (\ref{eq:priori-estimation}), (\ref{estimate:tildeu_M}), we obtain (\ref{eq:localization_err_estimation_inequalty}).
\end{proof}

\subsection{Error estimate for the time-discretization scheme}
We introduce a partition of time interval $[0,T]$ into subintervals $[\tau_{n-1}, \tau_n]$, $ n=1,2,\cdots, N$, such that $0=\tau_0<\tau_1<\cdots<\tau_N=T$. We define the length between $\tau_{n+1}$ and $\tau_n$ as $\Delta \tau \triangleq \tau_{n+1}-\tau_{n}$ and consider the pricing equation (\ref{eq:PIDE4SmoothedPayoff}) as follows
\begin{equation}\label{OperatorForm4SemiError}
u_{\tau}=\L[u] = \D[u]+\J[u]
\end{equation}
in a bounded domain (say, $\Omega_M$) satisfying homogeneous boundary conditions. Any error estimate obtained for a discretization of this equation can easily be applied to that of the localized problem (\ref{eq:truncated_call}) through a variable transformation $\bar{u}=u(\tau,x)-\tilde{g}(\tau,x)$.
For simplicity $\Delta \tau$ is assumed to be constant. Let $u^{n+1}$ be the solution of the following system resulting from Crank-Nicolson scheme
\begin{eqnarray}\label{eq:CNPIDE4NonsmoothPayoff}
\frac{u^{n+1}-u^n}{\triangle \tau}=\L \left[ \frac{u^{n+1}+u^n}{2}\right].
\end{eqnarray}
Defining the error function $e^n=u^n-u(\tau_n)$ and subtracting equation (\ref{eq:CNPIDE4NonsmoothPayoff}) from equation (\ref{OperatorForm4SemiError}) at $\tau_{n+\frac{1}{2}}$, we have the following error equation,
\begin{eqnarray}\label{timediscretization}
\frac{e^{n+1}-e^n}{\triangle \tau} -\L\left[ \frac{e^{n+1}+e^n}{2}\right] =-r_1^{n+\frac{1}{2}} + \L \left[ r_2^{n+\frac{1}{2}} \right],
\end{eqnarray}
where \begin{eqnarray*} \label{remainderterm}
r_1^{n+\frac{1}{2}} &=& \frac{u(\tau_{n+1})-u(\tau_n)}{\triangle \tau} - u_{\tau}(\tau_{n+\frac{1}{2}}),\\
r_2^{n+\frac{1}{2}} &=&\frac{u(\tau_{n+1})+u(\tau_n)}{2}-u(\tau_{n+\frac{1}{2}}).
\end{eqnarray*}
 We can estimate $r_1^{n+\frac{1}{2}}$ and $r_2^{n+\frac{1}{2}}$ via the Taylor expansion at $\tau_{n+\frac12}$ and its integral-form remainder:
\begin{eqnarray*}
 r_1^{n+\frac{1}{2}}
&=&\frac{1}{2 \triangle \tau} \left( \int_{\tau_{n}}^{\tau_{n+\frac{1}{2}}} u_{\tau\tau\tau}(\xi) (\xi-\tau_{n}) ^2 d \xi +\int_{\tau_{n+\frac{1}{2}}}^{\tau_{n+1}} u_{\tau\tau\tau}(\xi) (\tau_{n+1}-\xi) ^2 d \xi \right)\\
 r_2^{n+\frac{1}{2}} &=& \frac{1}{2} \left( \int_{\tau_{n}}^{\tau_{n+\frac{1}{2}}} u_{\tau\tau}(\xi) (\xi-\tau_{n}) d \xi +\int_{\tau_{n+\frac{1}{2}}}^{\tau_{n+1}} u_{\tau\tau}(\xi) (\tau_{n+1}-\xi) d \xi \right).
\end{eqnarray*}
Then  the Cauchy-Schwarz inequality gives  estimates of $r_1^{n+\frac{1}{2}}$ and $r_2^{n+\frac{1}{2}}$ as follows:
\begin{eqnarray*}
\norm{r_1^{n+\frac{1}{2}}}_{L^2(\Omega_M)}^2 &\leq& C (\triangle \tau)^3 \int_{\tau_n}^{\tau_{n+1}}\norm{u_{\tau\tau\tau}(\xi)}_{L^2(\Omega_M)}^2d\xi,\\
\norm{r_2^{n+\frac{1}{2}}}_{H^1(\Omega_M)}^2 &\leq& C (\triangle \tau)^3 \int_{\tau_n}^{\tau_{n+1}}\norm{u_{\tau\tau}(\xi)}_{H^1(\Omega_M)}^2 d\xi.
\end{eqnarray*}
Then we have
\begin{eqnarray}
\triangle \tau \sum_n \norm{r_1^{n+\frac{1}{2}}}_{L^2(\Omega_M)}^2 \leq C (\triangle \tau)^4 \int_{0}^{T} \norm{u_{\tau\tau\tau}(\xi )}_{L^2(\Omega_M)}^2 d\xi,\\
\triangle \tau \sum_n \norm{r_2^{n+\frac{1}{2}}}_{H^1(\Omega_M)}^2 \leq C (\triangle \tau)^4 \int_{0}^{T} \norm{u_{\tau\tau}(\xi )}_{H^1(\Omega_M)}^2  d\xi.
\end{eqnarray}

The variational formulation for error equation (\ref{timediscretization}) is given as follows:
find a series of $u^n\in H^1(\Omega_M),n=0,1,\ldots,N,$ such that $u^0=\tilde{h}(x)$ and for all $n=1,2,\ldots,N,$, for all v $\in H_0^1(\Omega_M)$,
\begin{eqnarray}\label{eq:VF4CN}
 \quad \left(\frac{e^{n+1}-e^n}{\triangle \tau},v\right)+a\left(\frac{e^{n+1}+e^n}{2},v\right)= -\left(r_1^{n+\frac{1}{2}},v\right) - a\left(r_2^{n+\frac{1}{2}},v\right),
\end{eqnarray}
where
\[a(u,v)=\int_{\Omega_M}\left( (\nabla u)^t \kappa ~ \nabla v -\alpha^t\nabla u ~v \right)dx-\int_{\Omega_M}\J[u]~v  dx.\]
Defining $\bar{e}^n=\frac{e^{n+1}+e^n}{2}$ and substituting $v$ with $e^{-2\beta \tau_{n+\frac{1}{2}}}\bar{e}^n$ in variational form (\ref{eq:VF4CN}), we have (seeing \garding inequality in Section 3.1)
\[a \left(\frac{e^{n+1}+e^n}{2},e^{-2\beta \tau_{n+\frac{1}{2}}}\bar{e}^n\right) \geq e^{-2 \beta \tau_{n+\frac{1}{2}}}\left(\underline{c} \norm{\bar{e}^n}_{H^1(\Omega_M)}^2 - \beta \norm{\bar{e}^n}_{L^2(\Omega_M)}^2 \right) \]
and
\begin{eqnarray*}
\left(\frac{e^{n+1}-e^n}{\triangle \tau},e^{-2\beta \tau_{n+\frac{1}{2}}}\bar{e}^n\right)
=\frac{1}{2\triangle \tau} \left(e^{-2\beta \tau_{n+1}}\norm{e^{n+1}}_{L^2}^2 -e^{-2\beta \tau_{n}}\norm{e^{n}}_{L^2}^2\right)+J,
\end{eqnarray*}
where (using $1-e^{-\beta \triangle \tau}\geq \beta \triangle \tau - \frac12 \beta^2 (\triangle \tau)^2$ and $e^{\beta \triangle \tau} -1 \geq \beta \triangle \tau + \frac12 \beta^2 (\triangle \tau)^2$)
\begin{eqnarray*}
J&=&\frac{1}{2\triangle \tau}(e^{n+1},e^{n+1}e^{-2\beta \tau_{n+\frac{1}{2}}}) (1-e^{- \beta \triangle \tau})
-\frac{1}{2\triangle \tau}(e^{n},e^{n}e^{-2\beta \tau_{n+\frac{1}{2}}}) (1-e^{\beta \triangle \tau})\\
&=&\frac{1}{2\triangle \tau}e^{-2\beta \tau_{n+\frac{1}{2}}}\left[(e^{n+1},e^{n+1}) (1-e^{- \beta \triangle \tau})+ (e^{n},e^{n}) (e^{\beta \triangle \tau}-1)\right]\\
&\geq &\beta e^{-2\beta \tau_{n+\frac{1}{2}}} \frac{\norm{e^{n+1}}_{L^2}^2 +\norm{e^{n}}_{L^2}^2}{2} +\frac{1}{4}\beta^2\triangle \tau \left(\norm{e^{n}}_{L^2}^2 - \norm{e^{n+1}}_{L^2}^2\right)
\\
&\geq & \beta e^{-2\beta \tau_{n+\frac{1}{2}}} \norm{\bar{e}^n}_{L^2}^2+\frac{1}{4}\beta^2\triangle \tau \left(\norm{e^{n}}_{L^2}^2 - \norm{e^{n+1}}_{L^2}^2\right).
\end{eqnarray*}
%
Thus
\begin{eqnarray*}
&&\left(\frac{e^{n+1}-e^n}{\triangle \tau},e^{-2\beta \tau_{n+\frac{1}{2}}}\bar{e}^n\right)+a\left(\frac{e^{n+1}+e^n}{2},e^{-2\beta \tau_{n+\frac{1}{2}}}\bar{e}^n\right)\\
&\geq& \frac{1}{2\triangle \tau}e^{-2\beta \tau_{n+1}}\norm{e^{n+1}}_{L^2}^2 -\frac{1}{2\triangle \tau}e^{-2\beta \tau_{n}}\norm{e^{n}}_{L^2}^2 \\
&&+ \underline{c}  e^{-2 \beta \tau_{n+\frac{1}{2}}}\norm{\bar{e}^n}_{H^1}^2+\frac{1}{4}\beta^2\triangle \tau \left(\norm{e^{n}}_{L^2}^2 - \norm{e^{n+1}}_{L^2}^2\right).
\end{eqnarray*}
We could also estimate for the right hand side of (\ref{eq:VF4CN}) by Young's inequality.
\begin{eqnarray*}
&&-\left(r_1^{n+\frac{1}{2}},e^{-2\beta \tau_{n+\frac{1}{2}}}\bar{e}^n\right) - a\left(r_2^{n+\frac{1}{2}},e^{-2\beta \tau_{n+\frac{1}{2}}}\bar{e}^n\right) \\
&\leq&  e^{-2\beta \tau_{n+\frac{1}{2}}}\norm{r_1^{n+\frac{1}{2}}}_{L^2} \norm{\bar{e}^n}_{L^2}  +\overline{c} e^{-2\beta \tau_{n+\frac{1}{2}}}\norm{r_2^{n+\frac{1}{2}}}_{H^1} \norm{\bar{e}^n}_{H^1}\\
&\leq& e^{-2\beta \tau_{n+\frac{1}{2}}}\left(\frac{1}{4\varepsilon} \norm{r_1^{n+\frac{1}{2}}}_{L^2}^2 + \varepsilon \norm{\bar{e}^n}_{L^2}^2\right) +\overline{c} e^{-2\beta \tau_{n+\frac{1}{2}}} \left(\frac{1}{4\varepsilon} \norm{r_2^{n+\frac{1}{2}}}_{H^1}^2 + \varepsilon \norm{\bar{e}^n}_{H^1}^2\right) \\
&=& \frac{1}{4\varepsilon}  e^{-2\beta \tau_{n+\frac{1}{2}}} \left(\norm{r_1^{n+\frac{1}{2}}}_{L^2}^2+\overline{c}\norm{r_2^{n+\frac{1}{2}}}_{H^1}^2\right) + \varepsilon (1+\overline{c}) e^{-2\beta \tau_{n+\frac{1}{2}}}\norm{\bar{e}^n}_{H^1}^2.
\end{eqnarray*}
Therefore the following inequality is obtained
\begin{eqnarray*}
&&e^{-2\beta \tau_{n+1}}\norm{e^{n+1}}_{L^2}^2 - e^{-2\beta \tau_{n}}\norm{e^{n}}_{L^2}^2 \\
&&+\frac{1}{4}\beta^2 (\triangle \tau)^2 \left(\norm{e^{n}}_{L^2}^2 - \norm{e^{n+1}}_{L^2}^2\right) + C_1 \triangle \tau e^{-2 \beta \tau_{n+\frac{1}{2}}}\norm{\bar{e}^n}_{H^1}^2 \\
&\leq&  \frac{1}{2\varepsilon}  e^{-2\beta \tau_{n+\frac{1}{2}}} \left(\triangle \tau\norm{r_1^{n+\frac{1}{2}}}_{L^2}^2+\overline{c} \triangle \tau \norm{r_2^{n+\frac{1}{2}}}_{H^1}^2\right) \\
&\leq& \frac{1}{2\varepsilon}  \left(\triangle \tau \norm{r_1^{n+\frac{1}{2}}}_{L^2}^2+\overline{c} \triangle \tau
\norm{r_2^{n+\frac{1}{2}}}_{H^1}^2 \right) ,
\end{eqnarray*}
where $C_1=2 \left(\underline{c} - \varepsilon (1+\overline{c})\right)$. Here we could choose $\varepsilon<\frac{\underline{c}}{1 +\overline{c}}$ to make $C_1>0$.

Summing the above inequality from $n=0$ to $n=m$ and combining estimates of $r_1^{n+\frac{1}{2}}$ and $r_2^{n+\frac{1}{2}}$, we have
\begin{eqnarray*}\label{time_error}
&&e^{-2\beta \tau_{m+1}}\norm{e^{m+1}}_{L^2}^2-e^{-2\beta \tau_{0}}\norm{e^{0}}_{L^2}^2 + \frac{1}{2}\beta^2\triangle \tau^2 \left(\norm{e^{0}}_{L^2}^2 - \norm{e^{m+1}}_{L^2}^2\right) \\
&&+ C_1 \sum_{n=0}^{m} e^{-2 \beta \tau_{n+\frac{1}{2}}}\norm{\frac{e^{n+1}+e^n}{2}}_{H^1}^2 \\
&\leq& \frac{1}{2\varepsilon}  \left(\triangle \tau \sum_{n=0}^{m}\norm{r_1^{n+\frac{1}{2}}}_{L^2}^2+\overline{c} \triangle \tau \sum_{n=0}^{m}\norm{r_2^{n+\frac{1}{2}}}_{H^1}^2\right)  \\
&\leq& C (\triangle\tau)^4 \int_0^T (\norm{u_{\tau\tau\tau}(\xi)}_{L^2(\Omega_M)}^2 + \norm{u_{\tau\tau}(\xi)}_{H^1(\Omega_M)}^2) d\xi.
\end{eqnarray*}
Here we note that the initial error $e^0=0$. The above estimate can be simplified as follows. If $\triangle \tau$ is sufficiently small so that $e^{-2\beta T} - \frac12 \beta^2 (\triangle \tau)^2 > 0$ then we have,
for all $m=0,1, \ldots, [T/\triangle \tau]$:
\begin{eqnarray}\label{time-error-estimate}
\nonumber && \norm{e^{m+1}}_{L^2(\Omega_M)} + \sum_{n=1}^{m+1} \norm{e^n}_{H^1(\Omega_M)}\\
& &\qquad \qquad \leq C (\triangle \tau)^2\left(\int_0^T (\norm{u_{\tau\tau\tau}(\xi)}_{L^2(\Omega_M)}^2 + \norm{u_{\tau\tau}(\xi)}_{H^1(\Omega_M)}^2) d\xi \right)^{\frac12},
\end{eqnarray}
where $C$ is a generic constant independent of $\triangle \tau$ but dependent on $T$ and the size of the domain $\Omega_M$. This gives an error estimate of the Crank-Nicolson scheme in any finite time interval.

\subsection{Error analysis to the finite element method}
 For the finite element spatial discretization, we introduce a triangular partition $\mathcal{T}_h$ of $\Omega_M$. We denote by $\mathcal{T}_h$ the set of all non-overlapping triangles $K$ of the partition and by $d(K)$ the longest side of $K$. Then
 denote $h=max_{K\in \mathcal{T}_h}d(K).$ We then construct a continuous piecewise polynomial finite element space $V_h$, i.e.,
\[V_h=\{v\in \mathbb{C}(\Omega_M): v|_K \in \mathcal{P}_k(K),\forall K\in \mathcal{T}_h,v|_{\partial \Omega_M}=0\},\]
where $\mathcal{P}_k(K)$ is the set of polynomial functions defined on $K$ with its highest order $k$. Based on the approximation property of \cite{Thomee:2006} for projection operator in FEM, we could have the similar approximation properties hold: for $u\in H_{\eta}^{k+1}(\Omega_{M})$ there is a projection $\P$ such that $\P u\in V_h$ and
\begin{eqnarray}\label{ProjectionApproximation}
\norm{u-\P u}_{H_{\eta}^s}\leq C h^{k+1-s} \norm{u}_{H_{\eta}^{k+1}}.
\end{eqnarray}
Note here that our domain is localized to a bounded one and weighted Sobolev spaces are not really necessarily. Nevertheless we just state results in terms of weighted Sobolev spaces where usual Sobolev spaces are special cases.
Applying $v \in V_h$ in the variational form (\ref{eq:truncated_variational_formulation}) in the bounded domain $\Omega_M$, we obtain the finite element solution $u_h$ satisfying:
\begin{eqnarray}\label{eq:truncated_variational_formulation_Vh} \qquad \left\{
\begin{array}{rcll}
(\frac{\partial u_h}{\partial \tau},v)_{L_{\eta}^2(\Omega_M)}+a_M^{\eta}(u_h,v)&=&\mathcal{A}[\tilde{g}],& \forall ~ v \in V_h\\
u_h(0,x)&=&0,& \forall ~ x \in \Omega_M\\
u_h(\tau,x) &=& 0, & \forall ~ (\tau,x)\in [0,T] \times \partial \Omega_M.
\end{array}
\right.
\end{eqnarray}
Now we want to estimate the error between the solution $u$ of (\ref{eq:truncated_variational_formulation}) and finite element solution $u_h$.
\begin{proposition} Let $e_h=u-u_h$ and $u$ have enough regularity, i.e. $u \in H_{\eta}^{2}(\Omega_M)$. Then for piecewise
 linear finite elements in $V_h$ we have the following error estimate:
\[\norm{e_h}_{L_{\eta}^2(\Omega_M)} +\left(\int_0^{\tau} \norm{e_h}_{H_{\eta}^1(\Omega_M)}^2d\xi \right)^{\frac12}  \leq C h
\left(\int_0^{\tau} \norm{u}_{H_{\eta}^2(\Omega_M)}^2 d\xi \right)^{\frac12}, \]
where $C$ is a generic constant depending on a given time $\tau$ and the size of $\Omega_M$ but not on $h$.
\end{proposition}
\begin{proof} Let $e_h = \theta + \phi$, where $\theta = u-\P u$, $\phi=\P u - u_h$ and the operator $\P$ is a projection from $H_{\eta}^1(\Omega_M)$ to $V_h$. Subtracting equation (\ref{eq:truncated_variational_formulation}) in \ref{ProjectionApproximation} from (\ref{eq:truncated_variational_formulation_Vh}), we have the error equation
\[\left(\frac{d e_h}{d \tau},v\right)_{L_{\eta}^2(\Omega_M)}+a_M^{\eta}(e_h,v)=0, \forall ~ v \in V_h\]
Let $v=\phi$. Since $(\theta,v)_{H_{\eta}^1(\Omega_M)}=0$ for all $v\in V_h$, we have
\begin{eqnarray*}
\frac{1}{2}\frac{d }{d \tau}\norm{\phi}_{L_{\eta}^2(\Omega_M)}^2+a_M^{\eta}(\phi,\phi) &=& - a_M^{\eta}(\theta,\phi).
\end{eqnarray*}
By \garding inequality (\ref{prop:garding_inequality}), we obtain
\begin{eqnarray*}
\frac{1}{2}\frac{d }{d \tau}\norm{\phi}_{L_{\eta}^2(\Omega_M)}^2 +  \underline{c} \norm{\phi}_{H_{\eta}^1(\Omega_M)}^2 - \beta \norm{\phi}_{L_{\eta}^2(\Omega_M)}^2  &\leq& \overline{c} \norm{\theta}_{H_{\eta}^1(\Omega_M)} \norm{\phi}_{H_{\eta}^1(\Omega_M)} \\
&\leq& \frac{\overline{c}}{4\varepsilon} \norm{\theta}_{H_{\eta}^1(\Omega_M)}^2 + \overline{c} \varepsilon \norm{\phi}_{H_{\eta}^1(\Omega_M)}^2
\end{eqnarray*}
Taking $\varepsilon = \frac{\underline{c}}{2 \overline{c}}$ and multiplying the above inequality by $e^{-2\beta \tau}$, we  have
\begin{eqnarray*}
e^{-2\beta \tau}\frac{d }{d \tau}\norm{\phi}_{L_{\eta}^2(\Omega_M)}^2 - 2 \beta e^{-2\beta \tau} \norm{\phi}_{L_{\eta}^2(\Omega_M)}^2
+ \underline{c} e^{-2 \beta \tau} \norm{\phi}_{H_{\eta}^1(\Omega_M)}^2 \leq \frac{\overline{c}^2}{\underline{c}} e^{-2\beta \tau} \norm{\theta}_{H_{\eta}^1(\Omega_M)}^2
\end{eqnarray*}
Noting the zero initial condition and integrating in $\tau$ from $\tau_0=0$ to $s$,  we have
\[\norm{\phi}_{L_{\eta}^2(\Omega_M)}^2 + \underline{c}\int_{\tau_0}^s e^{-2\beta (\tau-s)}\norm{\phi}_{H_{\eta}^1(\Omega_M)}^2 d\xi \leq \frac{\overline{c}^2}{\underline{c}} \int_{\tau_0}^s e^{-2\beta (\tau-s)}\norm{\theta}_{H_{\eta}^1(\Omega_M)}^2 d\xi .
\]
Thus combining it with the approximating property (\ref{ProjectionApproximation}) of the projection $\P$, we finish the proof.
\end{proof}

\section{Numerical Experiments}
In this section, we demonstrate the performance of the FEM and its related techniques proposed and analyzed in the paper. We take linear finite elements. The time step size is $\Delta \tau=0.01$. In the spatial discretization, we localize the computational domain $(S_1,S_2)$ in $[0,80] \times [0,80]$.
We will implement the integral term explicitly and the differential term implicitly for the pricing PIDE.

\subsection{Case study for a polynomial option}
\noindent  To ensure the accuracy of the finite element method for the multi-asset option pricing, we test the algorithm for a sample problem, i.e., the polynomial option with a payoff $(S_1+S_2)^2$, whose analytic solutions under the Black-Scholes model and Merton's Jump diffusion model are given in Appendix \ref{AnalyticSol4PolynomialOptions}. The parameters for the jump diffusion model and the polynomial option are given in Table \ref{Parameters4PolynomialOption}.

\begin{table}[htbp]
\centering
\begin{tabular}{lr}
\toprule
Parameters & Values \\
   \midrule
Diffusion volatility($\sigma_1,\sigma_2$) & $0.1-0.3, 0.1-0.3$ \\
Mean jump size($\nu_1,\nu_2$) & $-0.9,-0.9$\\
Mean jump volatility($\gamma_1,\gamma_2$) & $0.45,0.45$\\
Jump intensity($\lambda_1,\lambda_2$) & $0.1,0.1$\\
Correlation($\rho$)& $0.3$ or $-0.3$\\
\midrule
Underlying price ($S_1,S_2$) & $40,40$ \\
Strike price ($K$) & $80$\\
Interest rate ($r$) & $0.05$\\
Time to maturity ($T-t$) & $0.1$ or $0.9$\\
\bottomrule
\end{tabular}
 \caption{Parameters for polynomial option: $(S_1+S_2)^2$}
 \label{Parameters4PolynomialOption}
\end{table}
\begin{table}[htbp]
\centering
\begin{tabular}{cc|cc|ccc|c}
   \toprule
   $\tau$ & $\rho$ & $\sigma_1$  & $\sigma_2$ & BS Analytic & JD Analytic & FEM &  Relative Error  \\
   \midrule
   $0.1$ & $0.3$ & $0.1$  & $0.1$ & $6.4363$ & $6.4899$ &6.5695 & $1.23\%$ \\
   $0.1$ & $0.3$ & $0.1$  & $0.2$ & $6.4421$ & $6.4958$ &6.4785 & $0.27\%$ \\
   $0.1$ & $0.3$ & $0.1$  & $0.3$ & $6.4511$ & $6.5050$ &6.4434 & $0.95\%$ \\
   $0.1$ & $0.3$ & $0.2$  & $0.2$ & $6.4488$ & $6.5026$ &6.4977 & $0.08\%$ \\
   $0.1$ & $0.3$ & $0.2$  & $0.3$ & $6.4589$ & $6.5128$ &6.4611 & $0.79\%$ \\
   $0.1$ & $0.3$ & $0.3$  & $0.3$ & $6.4699$ & $6.5239$ &6.4646 & $0.91\%$ \\
      \midrule
   $0.9$ & $0.3$ & $0.1$  & $0.1$ & $6.7339$ & $7.2753$ &7.3561  & $1.11\%$ \\
   $0.9$ & $0.3$ & $0.1$  & $0.2$ & $6.7892$ & $7.3380$ &7.2909  & $0.64\%$ \\
   $0.9$ & $0.3$ & $0.1$  & $0.3$ & $6.8781$ & $7.4398$ &7.3828  & $0.77\%$ \\
   $0.9$ & $0.3$ & $0.2$  & $0.2$ & $6.8536$ & $7.5208$ &7.5093  & $0.15\%$ \\
   $0.9$ & $0.3$ & $0.2$  & $0.3$ & $6.9518$ & $7.6412$ &7.6183  & $0.30\%$ \\
   $0.9$ & $0.3$ & $0.3$  & $0.3$ & $7.0593$ & $7.4098$ &7.4093  & $0.01\%$ \\
   \bottomrule
\end{tabular}
\caption{Results for polynomial option with positive correlation: $\rho = 0.3$.}
\label{PositiveCoeff4PolynomialOption}
\end{table}
In Tables \ref{PositiveCoeff4PolynomialOption} and \ref{NegativeCoeff4PolynomialOption},  We report the prices for different volatilities of two underlying assets: $(0.1,0.1)$, $(0.1,0.2)$, $(0.2,0.2)$, $(0.2,0.3)$, $(0.3,0.3)$ , time to maturity $0.1$ or $0.9$, and correlation: positive (Table \ref{PositiveCoeff4PolynomialOption}) or negative
(Table \ref{NegativeCoeff4PolynomialOption}). We have computed these prices using both analytic and finite element methods in the square domain. The relative error is also provided. The prices computed from the FEM differ only slightly from the analytic solution. The difference is especially small in the finest triangulation.
The prices with the positive correlation are slightly higher than those with the negative correlation.
\begin{table}[htbp]
\centering
\begin{tabular}{cc|cc|ccc|c}
   \toprule
   $\tau$ & $\rho$ & $\sigma_1$  & $\sigma_2$ & BS Analytic & JD Analytic & FEM &  Relative Error  \\
     \midrule
   $0.1$ & $-0.3$ & $0.1$  & $0.1$ & $6.4343$ & $6.4880$ &6.5622 & $1.14\%$ \\
   $0.1$ & $-0.3$ & $0.1$  & $0.2$ & $6.4382$ & $6.4919$ &6.5573 & $1.01\%$ \\
   $0.1$ & $-0.3$ & $0.1$  & $0.3$ & $6.4453$ & $6.4992$ &6.4784 & $0.32\%$ \\
   $0.1$ & $-0.3$ & $0.2$  & $0.2$ & $6.4411$ & $6.4949$ &6.4699 & $0.39\%$ \\
   $0.1$ & $-0.3$ & $0.2$  & $0.3$ & $6.4473$ & $6.5012$ &6.4816 & $0.30\%$ \\
   $0.1$ & $-0.3$ & $0.3$  & $0.3$ & $6.4525$ & $6.5065$ &6.4540 & $0.81\%$ \\
   \midrule
   $0.9$ & $-0.3$ & $0.1$  & $0.1$ & $6.7158$ & $7.2572$ &7.1881 & $0.95\%$ \\
   $0.9$ & $-0.3$ & $0.1$  & $0.2$ & $6.7530$ & $7.3018$ &7.2686 & $0.46\%$ \\
   $0.9$ & $-0.3$ & $0.1$  & $0.3$ & $6.8239$ & $7.3855$ &7.3250 & $0.82\%$ \\
   $0.9$ & $-0.3$ & $0.2$  & $0.2$ & $6.7813$ & $7.3375$ &7.2680 & $0.95\%$ \\
   $0.9$ & $-0.3$ & $0.2$  & $0.3$ & $6.8433$ & $7.4124$ &7.4990 & $1.17\%$ \\
   $0.9$ & $-0.3$ & $0.3$  & $0.3$ & $6.8966$ & $7.4785$ &7.4422 & $0.49\%$ \\
   \bottomrule
\end{tabular}
\caption{Results for polynomial option with the negative correlation: $\rho = -0.3$}\label{NegativeCoeff4PolynomialOption}
\end{table}

The prices for the Black-Scholes model are lower than its jump diffusion counterpart. This is because the jump component explains some volatility for the underlying in addition to the diffusion volatility. We can see that the jump volatility explains around $10 \%$ volatility in terms of prices. The solution surface for the polynomial option with jump diffusion model is given in Figure \ref{Fig:JD2DEuro-Polynomial}.
\begin{figure}[htbp]
      \centering
      \includegraphics[width=0.7\textwidth]{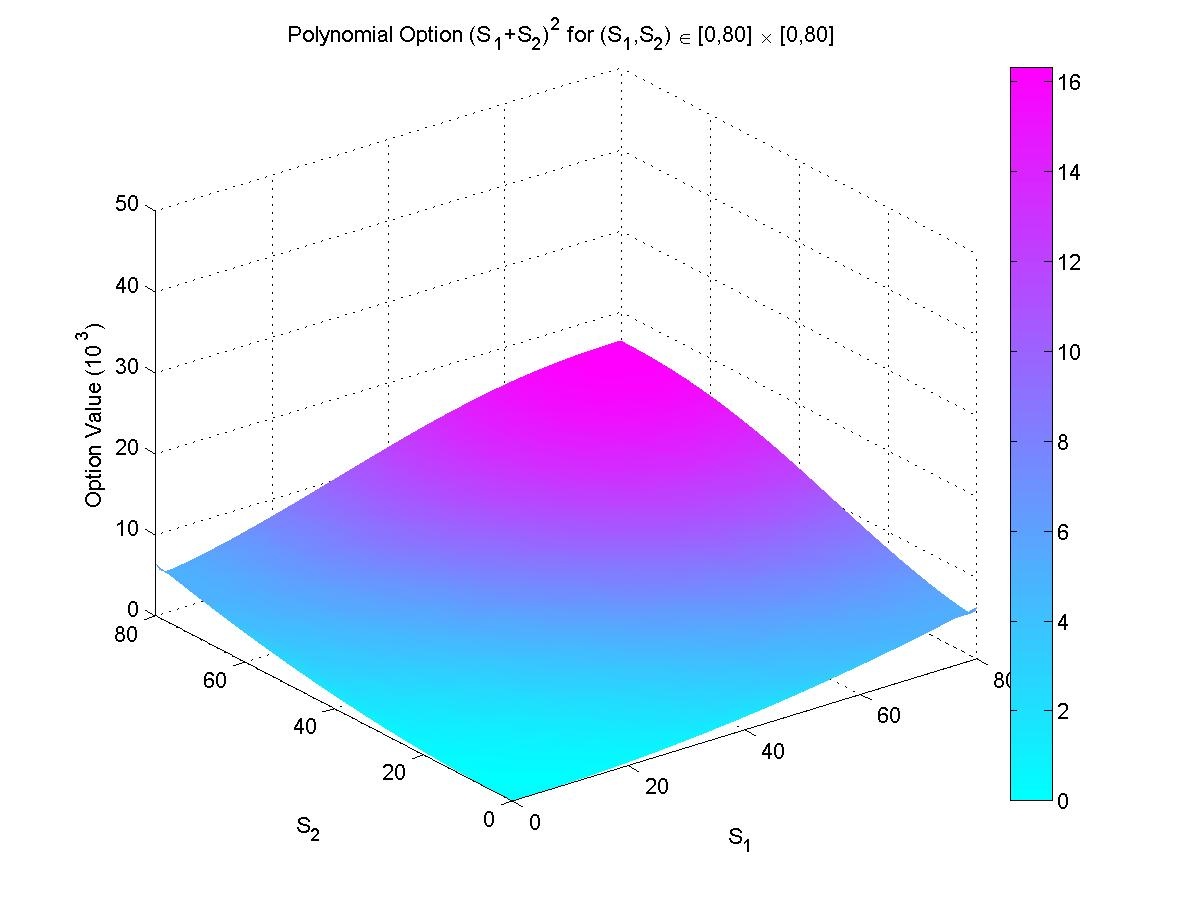}
      \caption{Polynomial option under Merton's jump diffusion model}
      \label{Fig:JD2DEuro-Polynomial}
\end{figure}

\begin{remark}
In Tables \ref{PositiveCoeff4PolynomialOption} and \ref{NegativeCoeff4PolynomialOption}, the prices is in unit of $1000$ and $\tau$ is the time to maturity. ``BS Analytic"  and ``JD Analytic" are analytic solutions provided in Appendix \ref{AnalyticSol4PolynomialOptions}. FEM is the solution computed by the FEM  and the relative error is defined as $\frac{|FEM-JD_Analytic|}{JD_Analytic}$.
\end{remark}

\subsection{Case study for other multi-asset options}
\noindent The basket put option is similar to a plain vanilla option except that the underlying is replaced by the weighted sum of the assets composing the basket. The payoff of a basket put option with positive weights $(w_1,w_2,\cdots,w_d)$ and strike price $K$ is given by
\[\left(K-\sum_{i=1}^{d}w_iS_i(T)\right)^+ .\]
The parameters for the underlying dynamics and the basket put option are given in Table \ref{Paramters4BasketPut}. Those model coefficients are of a magnitude that would be plausible in a real market.
\begin{table}
\centering
\begin{tabular}{lr}
\toprule
Parameters & Values \\
   \midrule
Diffusion volatility ($\sigma_1,\sigma_2$) & $0.3,0.3$ \\
Mean jump size ($\nu_1,\nu_2$) & $-0.9,-0.9$\\
Mean jump volatility ($\gamma_1,\gamma_2$) & $0.45,0.45$\\
Jump intensity ($\lambda_1,\lambda_2$) & $0.1,0.1$\\
Correlation ($\rho$)& $0.3$\\
\midrule
Underlying price ($S_1,S_2$) & $40,40$ \\
Weights ($w_1,w_2$) & $0.5,0.5$ \\
Strike price ($K$) & $40$\\
Interest rate ($r$) & $0.05$\\
Time to maturity ($T-t$) & $1$\\
\bottomrule
\end{tabular}
 \caption{Parameters for a basket put option: $\left(K-\sum_{i=1}^2 w_iS_i\right)^+$}
 \label{Paramters4BasketPut}
\end{table}

In order to facilitate the comparison with plain vanilla options, it is convenient to decompose the strike price $K$ as a function of the weights $w_i$, the initial prices of underlying and a vector of parameters $k_i$ that can be interpreted as indicators of moneyness of the plain vanilla options on the underlying assets:
\[K=\sum_{i=1}^dw_i\cdot k_i \cdot S_i(0).\]
If all the parameters $k_i$ are equal to 1, i.e., the individual plain vanilla options are at the money, the payoff of the basket put option can be rewritten as
\[\left(\sum_{i=1}^{d}w_i (S_i(0) - S_i(T)) \right)^+ .\]

From Jensen's inequality it follows that
\[\left( \sum_{i=1}^{d}w_i(S_i(0)-S_i(T))\right)^+ \leq \sum_{i=1}^{d} w_i \left( S_i(0)-S_i(T)\right)^+,\]
which leads us to the conclusion that a basket put option will be always cheaper than the portfolio of plain vanillas (with weights $w_i$) written on the same underlying assets.



In Table \ref{Results4Multi-AssetOption} we report the prices of basket put option, maximum/minimum of 2 put options for different volatilities of two underlying assets: $(0.1,0.1)$, $(0.1,0.2)$, $(0.2,0.2)$, $(0.2,0.3)$, $(0.3,0.3)$ , time to maturity $0.1,0.9$ and positive correlation $0.3$.
\begin{table}[htbp]
\centering
\begin{tabular}{ccc|c|c|c}
   \toprule
   $\tau$  & $\sigma_1$  & $\sigma_2$ & Basket Put  & Max of 2 Put & Min of 2 Put \\
     \midrule
   $0.1$  & $0.1$  & $0.1$ & $1.8015$ & $1.7997$ &1.8038 \\
   $0.1$  & $0.1$  & $0.2$ & $1.8342$ & $1.8329$ &1.8389 \\
   $0.1$  & $0.1$  & $0.3$ & $1.9127$ & $1.9096$ &1.9161 \\
   $0.1$  & $0.2$  & $0.2$ & $1.8857$ & $1.8806$ &1.8901 \\
   $0.1$  & $0.2$  & $0.3$ & $1.9834$ & $1.9794$ &1.9873 \\
   $0.1$  & $0.3$  & $0.3$ & $2.0723$ & $2.0702$ &2.0782 \\
   \midrule
   $0.9$  & $0.1$  & $0.1$ & $0.6059$ & $0.6012$ &0.7002  \\
   $0.9$  & $0.1$  & $0.2$ & $1.2413$ & $1.2392$ &1.2485 \\
   $0.9$  & $0.1$  & $0.3$ & $1.9280$ & $1.9231$ &1.9317  \\
   $0.9$  & $0.2$  & $0.2$ & $1.7769$ & $1.7700$ &1.7810  \\
   $0.9$  & $0.2$  & $0.3$ & $2.4397$ & $2.4352$ &2.4427 \\
   $0.9$  & $0.3$  & $0.3$ & $3.0011$ & $2.9923$ &3.0123  \\
   \bottomrule
\end{tabular}
\caption{Results for some two-asset put options}
\label{Results4Multi-AssetOption}
\end{table}
For positive weights $w_1,w_2$ satisfying $w_1+w_2=1$ we have
\[ (K-\max(S_1, S_2))^+ \leq (K-(w_1 S_1+w_2 S_2))^+  \leq (K-\min(S_1, S_2))^+ .\]
This implies that at any time, the price of basket put option is bigger than maximum of 2 put option and smaller than minimum of 2 put option, as shown in Table \ref{Results4Multi-AssetOption}.
The solution surfaces of basket put option, maximum of two put option and minimum of two put option are given in Figures \ref{Fig:JD2DEuro-BasketPut}, \ref{Fig:JD2DEuro-MaxPut} and \ref{Fig:JD2DEuro-MinPut}.

\begin{figure}[htbp]
      \centering
      \includegraphics[width=0.7\textwidth]{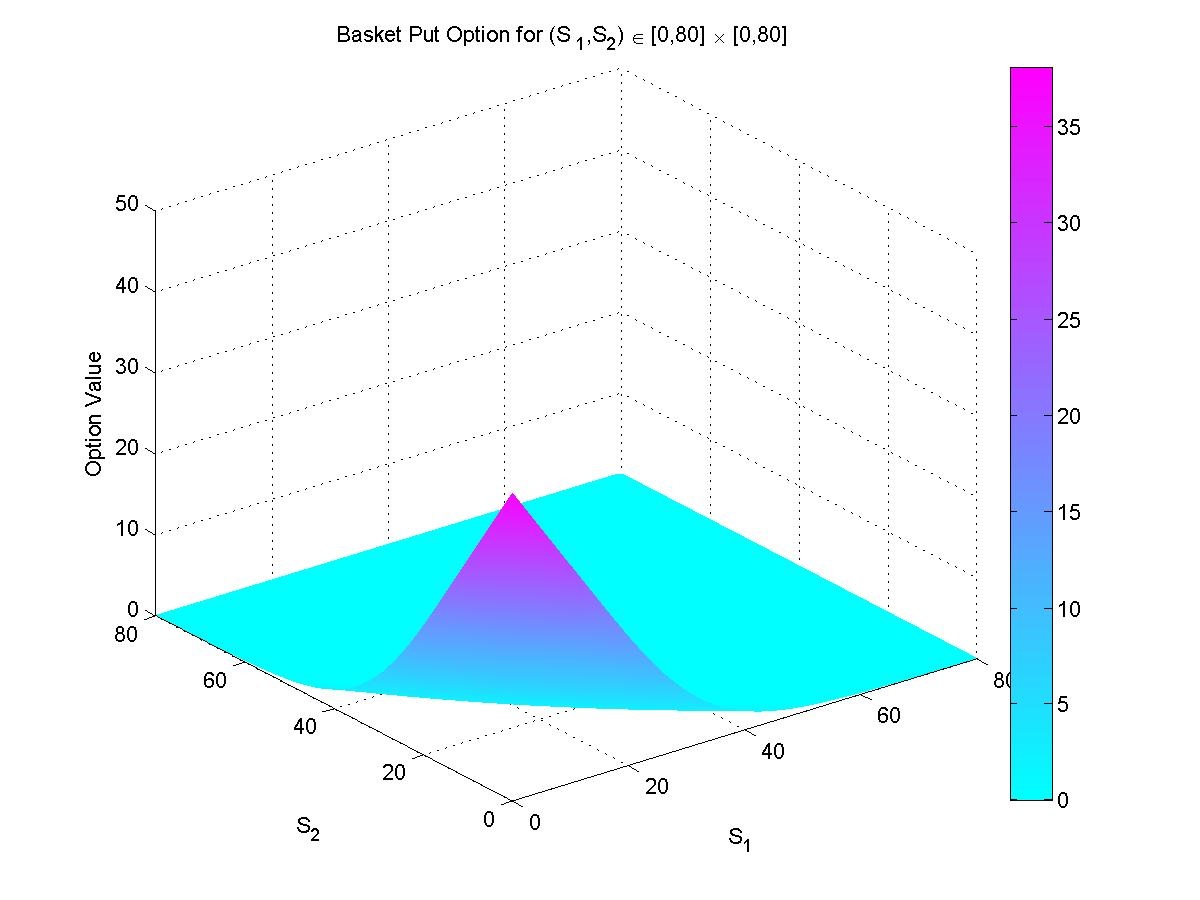}
      \caption{Basket put option under Merton's jump diffusion model}
      \label{Fig:JD2DEuro-BasketPut}
\end{figure}

\begin{figure}[htbp]
      \centering
      \includegraphics[width=0.7\textwidth]{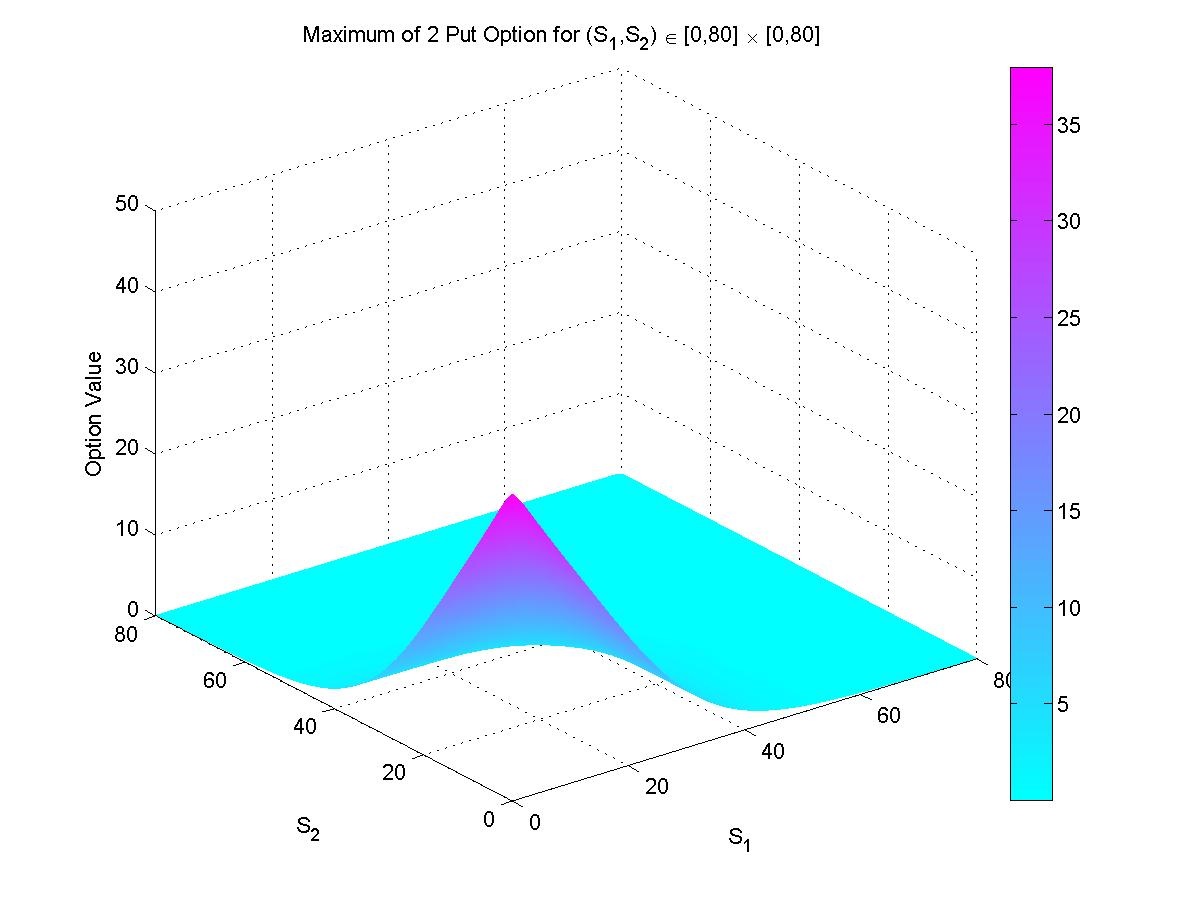}
      \caption{Maximum of 2 put option under Merton's jump diffusion model}
      \label{Fig:JD2DEuro-MaxPut}
\end{figure}

\begin{figure}[htbp]
      \centering
      \includegraphics[width=0.7\textwidth]{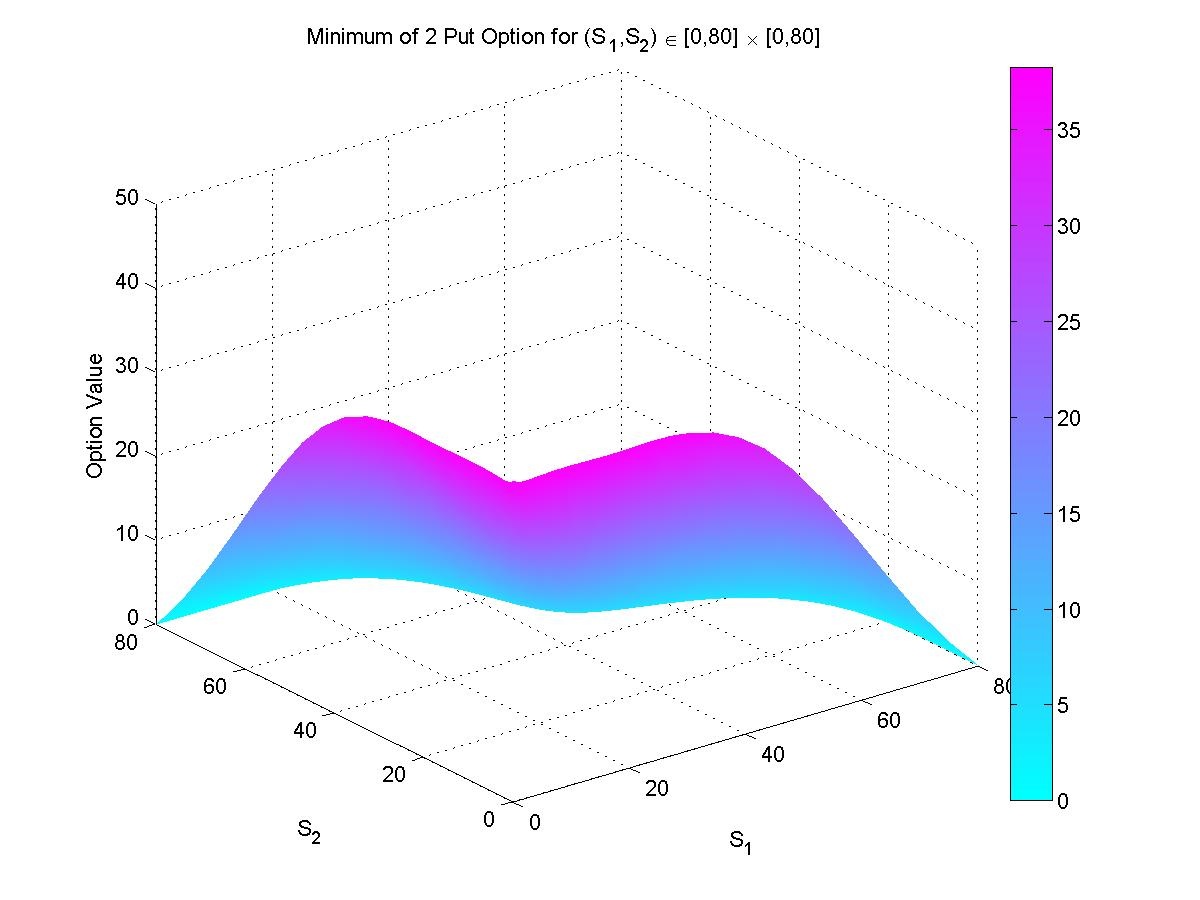}
      \caption{Minimum of 2 put option under Merton's jump diffusion model}
      \label{Fig:JD2DEuro-MinPut}
\end{figure}
\section*{Acknowledgments}
The authors are grateful to Olivier Pironneau for introducing their software \emph{FreeFem++} which simplifies some of the numerical experiments and Xiliang Lu for his valuable discussion in part of the numerical analysis of this work.

\appendix
\section{Property for $\eta(x)$}\label{AppendixA}
Suppose
\begin{eqnarray}\label{eta1}
\eta(x)=\left\{
\begin{array}{ll}
\eta_{1} (|x_1|+ |x_2|)& \textrm{if $x_1<0,x_2<0$,}\\
\eta_{1} |x_1|+ \eta_2 |x_2| & \textrm{if $x_1<0,x_2>0$,}\\
\eta_{2} |x_1|+ \eta_1 |x_2| & \textrm{if $x_1>0,x_2<0$,}\\
\eta_{2} (|x_1|+ |x_2|)& \textrm{if $x_1>0,x_2>0$.}
\end{array}
\right.
\end{eqnarray}
where $\eta_1>1$ and $\eta_2 >1$.
Then $\eta(x)$ satisfies $\eta\in L_{loc}^1(\mathbb{R})$,
$\nabla \eta \in L^{\infty}(\mathbb{R})$ and
\begin{eqnarray}
\Delta_{\theta}\eta=\eta(x+\theta y)-\eta(x)\leq \eta(y), \quad
\forall x,y \in \mathbb{R}^2,\norm{\theta}\leq 1.
\end{eqnarray}
We discuss it in the following $16(=2^4)$ cases:
\begin{itemize}
    \item [Case 1.] If $x_i+\theta y_i>0, x_i>0$ for $i=1,2$, then
    $\Delta_{\theta}\eta=\eta_2\theta(y_1+y_2)$, thus
\begin{displaymath}
\Delta_{\theta}\eta<\left\{
\begin{array}{ll}
-\eta_1(y_1+y_2)& \textrm{if $y_1<0,y_2<0$,}\\
-\eta_1 y_1+\eta_2 y_2 & \textrm{if $y_1<0,y_2>0$,}\\
\eta_2 y_1 - \eta_1 y_2 & \textrm{if $y_1>0,y_2<0$,}\\
\eta_2(y_1+y_2) & \textrm{if $y_1>0,y_2>0$.}
\end{array}
\right.
\end{displaymath}
which means $\Delta_{\theta}\eta<\eta(y)$.
    \item [Case 2.] If $x_i+\theta y_i>0$ for $i=1,2$ and $x_1>0,x_2<0$, then
\[\Delta_{\theta}\eta=\eta_2\theta(y_1+y_2)+(\eta_1+\eta_2)x_2\leq \eta_2\theta(y_1+y_2).\]
Note that $y_2>0$ in this case and
\begin{displaymath}
\eta(y)=\left\{
\begin{array}{ll}
-\eta_1y_1+\eta_2y_2& \textrm{if $y_1<0$,}\\
\eta_2 (y_1 + y_2) & \textrm{if $y_1>0$.}
\end{array}
\right.
\end{displaymath}
Thus $\Delta_{\theta}\eta<\eta(y)$.
    \item [Case 3.] If $x_i+\theta y_i>0$ for $i=1,2$ and $x_1<0,x_2>0$, then
    \[\Delta_{\theta}\eta=\eta_2\theta(y_1+y_2)+(\eta_1+\eta_2)x_1\leq \eta_2\theta(y_1+y_2)\leq \eta_2(y_1+y_2).\]
    Note that $y_1>0$ in this case and
\begin{displaymath}
\eta(y)=\left\{
\begin{array}{ll}
\eta_2y_1-\eta_1y_2& \textrm{if $y_2<0$,}\\
\eta_2 (y_1 +  y_2) & \textrm{if $y_2>0$.}
\end{array}
\right.
\end{displaymath}
Thus $\Delta_{\theta}\eta<\eta(y)$.
    \item [Case 4.] If $x_i+\theta y_i>0, x_i<0$ for $i=1,2$, then
    \[\Delta_{\theta}\eta=\eta_2\theta(y_1+y_2)+(\eta_1+\eta_2)(x_1+x_2)\leq \eta_2\theta(y_1+y_2)\leq \eta (y_1+y_2). \]
    Note that $y_1>0,y_2>0$ in this case, which implies
    $\eta(y)=\eta_2(y_1+y_2)$.

Thus $\Delta_{\theta}\eta<\eta(y)$.

\item [Case 5-16.] These remaining $12$ cases can be grouped by changing the signs of $x_i+\theta y_i$ for $i=1,2$. They are similar to the above four cases. we
omit here.
\end{itemize}
More generally, we can get the following result:
\begin{eqnarray}
\pm\big(\eta(x+\theta y)-\eta(x)\big)\leq \eta(\pm y), \quad \forall x,y \in
\mathbb{R}^2,\norm{\theta}\leq 1.
\end{eqnarray}

\section{Analytic solution for polynomial option}\label{AnalyticSol4PolynomialOptions}
Let's consider options whose payoff is a polynomial function of the underlying price at expiration; so-called \emph{polynomial option}. Here we focus on the polynomial option with payoff $H(S_1,S_2)= (S_1+S_2)^2$, which is a differentiable function w.r.t $S_1,S_2$. Before deriving the exact solution of PIDE to pricing the polynomial options under Exponential \levy models, let's work in the well-known Black-Scholes-Merton framework. We obtain that the closed-form formula $V_{BS}(t,S_1,S_2)$ for the polynomial option solves the following two dimensional BS-type PDE:
\begin{eqnarray}\label{eq:BS4PowerOption}
\left\{
\begin{array}{l}
\dsp\frac{\partial V}{\partial t}+ r S_1\frac{\partial V}{\partial S_1} + rS_2 \frac{\partial V}{\partial S_2}-rV \\[2mm]
\dsp \qquad + \frac{1}{2}\sigma_1^2 S_1^2 \frac{\partial^2 V}{\partial S_1^2} + \sigma_1\sigma_2 \rho S_1 S_2\frac{\partial^2 V}{\partial S_1\partial S_2} + \frac{1}{2}\sigma_2^2S_2^2\frac{\partial^2 V}{\partial S_2^2} = 0, \\ [2mm]
V(T,S_1, S_2) = V_{BS}(T,S_1,S_2) = (S_1+S_2)^2, \\ [2mm]
V(t,S_1, 0) =  S_1^2 e^{(r+\sigma_1^2)(T-t)}, \quad  t<T, \\ [2mm]
V(t,0, S_2) = S_2^2 e^{(r+\sigma_2^2)(T-t)}, \quad  t<T, \\ [2mm]
V(t,S_1, S_2) = V_{BS}(t,S_1,S_2), \quad \mbox{if } S_1,S_2 \rightarrow \infty \mbox{ and } t<T,
\end{array}\right.
\end{eqnarray}
where
\begin{equation}\label{sol:BS4PowerOption}
V_{BS}(t,S_1,S_2)=S_1^2 e^{(r+\sigma_1^2)(T-t)}+ S_2^2 e^{(r+\sigma_2^2)(T-t)} + 2 S_1 S_2 e^{(r+\rho\sigma_1\sigma_2)(T-t)}.
\end{equation}

If the polynomial option is priced under exponential \levy model, then the pricing equation becomes:
\begin{eqnarray}\label{eq:Levy4PowerOption}
\nonumber & &\frac{\partial V}{\partial t} + rS_1\frac{\partial V}{\partial S_1}+rS_2\frac{\partial V}{\partial S_2} + \frac{1}{2}\sigma_1^2S_1^2\frac{\partial ^2 V}{\partial S_2^2} + \rho \sigma_1 \sigma_2 S_1 S_2\frac{\partial^2 V}{\partial
S_1\partial S_2}+\frac{1}{2}\sigma_2^2S_2^2\frac{\partial ^2 V}{\partial S_2^2}- rV\\
& & \qquad \qquad + \int_{\mathbb{R}}\big(V(t,S_1e^x,S_2)-V(t,S_1,S_2)-S_1(e^x-1) \frac{\partial V}{\partial S_1}(t,S_1,S_2)\big)\nu_1(dx)\\
\nonumber & & \quad + \int_{\mathbb{R}}\big(V(t,S_1,S_2e^x)-V(t,S_1,S_2)-S_2(e^x-1) \frac{\partial V}{\partial S_2}(t,S_1,S_2)\big)\nu_2(dx)=0,
\end{eqnarray}
where $\nu_i(dx)=\lambda_i e^xp_i(e^x)dx, \lambda_1\neq 0, \lambda_2 \neq 0.$

From the exact solution (\ref{sol:BS4PowerOption}) of polynomial option under the Black-Scholes model, we know $S_1^2e^{(r+\sigma_1^2)(T-t)}, S_2^2e^{(r+\sigma_2^2)(T-t)}, S_1S_2e^{(r+\rho\sigma_1\sigma_2)(T-t)}$ are all functions satisfying the first equation of (\ref{eq:BS4PowerOption}). So it is not feasible to use their linear combination to construct the exact solution of
(\ref{eq:Levy4PowerOption}).

Since the jump component of  (\ref{eq:Levy4PowerOption}) will increase the volatility of the stock price, it is natural to assume that the solution of (\ref{eq:Levy4PowerOption}) is the following:
\begin{equation}\label{sol:Levy4PowerOption}
V_{ELM}(t,S)=S_1^2e^{(r+ a \sigma_1^2)(T-t)} + S_2^2e^{(r+b\sigma_2^2)(T-t)} + 2S_1S_2e^{(r+ c \rho\sigma_1\sigma_2)(T-t)},
\end{equation}
where constants $a,b,c$ are to be specified.

Substituting this solution into equation (\ref{eq:Levy4PowerOption}) and rearranging the terms, the non-jump terms
can be rewritten as:
\begin{eqnarray}\label{secondorder term}
&& \frac{\partial V}{\partial t}+rS_1\frac{\partial V}{\partial S_1}+rS_2\frac{\partial V}{\partial S_2} - rV \\
\nonumber && + \frac{1}{2}\sigma_1^2S_1^2\frac{\partial ^2 V}{\partial S_2^2} + \rho \sigma_1 \sigma_2 S_1 S_2\frac{\partial^2 V}{\partial S_1 \partial S_2} + \frac{1}{2}\sigma_2^2 S_2^2\frac{\partial ^2 V}{\partial S_2^2} \\
\nonumber &=& \dsp(1-a)\sigma_1^2S_1^2e^{(r+a\sigma_1^2)(T-t)}+(1-b)\sigma_2S_2^2e^{(r+b\sigma_2^2)(T-t)} \\
\nonumber&& + 2(1-c)\sigma_1\sigma_2S_1S_2e^{(r+c\rho\sigma_1\sigma_2)(T-t)}.
\end{eqnarray}
The jump term in (\ref{eq:Levy4PowerOption}) can also be rearranged as follows:
\begin{eqnarray}\label{jumpterm}
&& \int_{\mathbb{R}}\big(C(t,S_1e^x,S_2)-C(t,S_1,S_2)-S_1(e^x-1) \frac{\partial C}{\partial S_1}(t,S_1,S_2)\big)\nu_1(dx)\\
\nonumber && +\int_{\mathbb{R}}\big(C(t,S_1,S_2e^x)-C(t,S_1,S_2)-S_2(e^x-1) \frac{\partial C}{\partial S_2} (t,S_1,S_2)\big) \nu_2(dx)\\
\nonumber &=&  \lambda_1 \int_{\mathbb{R}}(e^x-1)^2 S_1^2e^{(r+a\sigma_1^2)(T-t)}e^xp_1(e^x)dx \\
\nonumber && + \lambda_2 \int_{\mathbb{R}}(e^x-1)^2 S_2^2e^{(r+b\sigma_2^2)(T-t)}e^xp_2(e^x)dx\\
\nonumber &=&  \Lambda_1 S_1^2e^{(r+a\sigma_1^2)(T-t)}+\Lambda_2 S_2^2e^{(r+b\sigma_2^2)(T-t)},
\end{eqnarray}
where $\Lambda_i=\lambda_i \int_{\mathbb{R}^+} (y-1)^2p_i(y)dy=\lambda_i \big(e^{2(\nu_i+\gamma_i^2)} - 2e^{(\nu_i+\frac{1}{2}\gamma_i^2)} - 1\big),i=1,2.$

Thus combining (\ref{secondorder term}) and (\ref{jumpterm}) together we have :
\begin{eqnarray*}
\left \{
\begin{array}{rcl}
(1-a)\sigma_1^2+ \Lambda_1 &=&0\\
(1-b)\sigma_2^2+ \Lambda_2 &=&0\\
2(1-c)\sigma_1\sigma_2+0 &=&0
\end{array}
\right.
\end{eqnarray*}
and then
\[a \sigma_1^2=\sigma_1^2+\Lambda_1,b\sigma_2^2=\sigma_2^2+\Lambda_2,c=1.\]
Therefore, we can obtain the exact solution (\ref{sol:Levy4PowerOption}) of (\ref{eq:Levy4PowerOption}):
\begin{eqnarray}
V_{ELM}(t,S_1,S_2) &=& S_1^2 e^{(r+\sigma_1^2+ \Lambda_1)(T-t)}+ S_2^2 e^{(r+\sigma_2^2+\Lambda_2)(T-t)}\\
\nonumber          & & + 2 S_1 S_2 e^{(r+ \rho\sigma_1\sigma_2)(T-t)}
\end{eqnarray}
with the following boundary conditions:
\begin{eqnarray}
\qquad \left \{
\begin{array}{rcl}
V(T,S_1,S_2)&=& (S_1+S_2)^2,\\[2mm]
V(t,S_1,0)  &=& S_1^2 e^{(r+\sigma_1^2+ \Lambda_1)(T-t)}, \quad t<T, \\[2mm]
V(t,0,S_2)  &=& S_2^2 e^{(r+\sigma_2^2+\Lambda_2)(T-t)}, \quad t<T, \\[2mm]
V(t,S_1, S_2) &=& V_{ELM}(t,S_1,S_2), \quad S_1,S_2 \rightarrow \infty \mbox{ and } t<T.
\end{array}
\right.
\end{eqnarray}
By the variable transformation (\ref{eq:var_transforamtion}):
\[x_i=\ln S_i,\tau=T-t,V(\tau,x)=e^{r\tau}C(T-\tau,e^{x_1},e^{x_2}),\]
it is easy to derive the solution $V(\tau,x)$ for the following pricing PIDE:
\begin{eqnarray}
\qquad \qquad \left\{\begin{array}{rcl} \dsp\frac{\partial V}{\partial
\tau}(\tau,x)&=&\nabla\cdot (\kappa \nabla V)+ \nabla \cdot
(\alpha V)+\mathcal{I}[V](\tau,x),\\[2mm]
V(0,x) &=& (e^{x_1}+e^{x_2})^2, \quad (x)\in \Omega,\\[2mm]
V(\tau,x) &=& e^{2x_1+(2r+\sigma_1^2+\Lambda_1)\tau}, \quad
x_2 \rightarrow -\infty,\tau \in
(0,T]\\[2mm]
V(\tau,x)&=& e^{2x_2+(2r+\sigma_2^2+ \Lambda_2)\tau},\quad x_1 \rightarrow -\infty,\tau \in (0,T] \\[2mm]
V(\tau,x) &=& V_{ELM}(\tau,x), \quad \mbox{if } x_1,x_2 \rightarrow \infty,\tau\in (0,T].
\end{array}\right.
\end{eqnarray}
Its exact solution is:
\[V_{ELM}(\tau,x)=e^{2x_1+(2r+\sigma_1^2+ \Lambda_1)\tau}+
e^{2x_2+(2r+\sigma_2^2+\Lambda_2)\tau} +2e^{x_1+x_2+(2r+
\rho\sigma_1\sigma_2)\tau}.\]

\end{document}